\documentclass[final,3p,times]{elsarticle}
\usepackage{amsfonts}
\usepackage{subfigure}
\usepackage{amssymb}
\usepackage{amsmath}
\usepackage{graphics}
\usepackage{epsfig}
\usepackage{epstopdf}
\usepackage{amsthm}
\usepackage{textcomp}
\usepackage{color}
\usepackage{float}
\usepackage{multirow}
\biboptions{sort&compress}

\usepackage{algorithmic}
\usepackage{algorithm}

\journal{}

\begin{document}

\begin{frontmatter}

\title{Security evaluation of cyber networks under advanced persistent threats}

\cortext[cor1]{Corresponding author}

\author[rvt,rvt2]{Lu-Xing Yang}
\ead{ylx910920@gmail.com}

\author[rvt3]{Pengdeng Li}
\ead{1414797521@qq.com}

\author[rvt3]{Xiaofan Yang\corref{cor1}}
\ead{xfyang1964@gmail.com}

\author[rvt]{Luosheng Wen}
\ead{wls@cqu.edu.cn}

\author[rvt3]{Yingbo Wu}
\ead{wyb@cqu.edu.cn}

\author[rvt4]{Yuan Yan Tang}
\ead{yytang@umac.mo}

\address[rvt]{School of Mathematics and Statistics, Chongqing University, Chongqing, 400044, China}

\address[rvt2]{Faculty of Electrical Engineering, Mathematics and Computer Science, Delft University of Technology,  Delft, GA 2600, The Netherlands}

\address[rvt3]{School of Software Engineering, Chongqing University, Chongqing, 400044, China}

\address[rvt4]{Department of Computer and Information Science, The University of Macau, Macau}

\begin{abstract}
	
This paper is devoted to measuring the security of cyber networks under advanced persistent threats (APTs). First, an APT-based cyber attack-defense process is modeled as an individual-level dynamical system. Second, the dynamic model is shown to exhibit the global stability. On this basis, a new security metric of cyber networks, which is known as the limit security, is defined as the limit expected fraction of compromised nodes in the networks. Next, the influence of different factors on the limit security is illuminated through theoretical analysis and computer simulation. This work helps understand the security of cyber networks under APTs.
	
\end{abstract}


\end{frontmatter}



\section{Introduction}

Cyberspace has come to be an integral part of our society. Government agencies, schools, hospitals, corporations, financial institutions and other organizations ceaselessly collect, process, and store a great deal of data on computers and transmit these confidential data across networks to other computers \cite{Kitchin1998, Dodge2000}. However, cyberspace is vulnerable to a wide range of cyber threats. Sophisticated cyber perpetrators exploit vulnerabilities to steal information and money and develop capabilities to disrupt and destroy essential cyber services. In light of the risk and potential consequences of cyber attacks, strengthening the security and resilience of cyberspace has become an important mission. Cybersecurity is committed to protecting computers, networks, programs and data from unintended or unauthorized access, change, or destruction \cite{Shoemaker2011, Kostopoulos2012, Singer2014}. You cannot manage if you cannot measure. Before working out cyberspace security solutions, the security of cyber networks must be evaluated \cite{Jaquith2007, Jensen2009, Cheng2014}.   

Advanced persistent threats (APTs) are a newly emerging class of cyber attacks. With a clear goal, an APT attack is highly targeted, well-organized, well-resourced, covert and long-term \cite{Tankard2011, ChenP2014, HuPF2015}. APTs pose a severe threat to cyberspace, because they invalidate conventional cyber defense mechanisms. In the last decade, the number of APTs increased rapidly and numerous security incidents were reported all over the world \cite{Rass2017}. For the purpose of resisting APTs, it is vital to evaluate the security of cyber networks under APTs. However, due to the persistence of APTs, existing security evaluation methods are not applicable to APTs \cite{Phillips1998, Kotenko2006, Frigault2008, Lippmann2012, Yusuf2017}. Recently, Pendleton et al. \cite{Pendleton2017} considered the expected fraction of compromised nodes in a cyber network as a security metric of the network. As the fraction is varying over time, its availability is questionable. 

To measure the security of a cyber network under APTs, an APT-based cyber attack-defense process must be modeled as a continuous-time dynamical system. The individual-level dynamical modeling technique, which has been applied to areas such as epidemic spreading \cite{Mieghem2009, Mieghem2011, Sahneh2012}, malware spreading \cite{XuSH2012a, XuSH2012b, XuSH2014, YangLX2015, YangLX2017a, YangLX2017b, WuYB2017, YangLX2017c}, rumor spreading \cite{YangLX2017d, YangLX2017e} and viral marketing \cite{ZhangTR2017b}, is especially suited to the modeling of APT-based cyber attack-defense processes, because the topological structure of the targeted cyber network can be accommodated \cite{XuSH2014a}. Towards this direction, a number of APT-based cyber attack-defense models have been proposed \cite{LuWL2013, XuSH2015a, ZhengR2015}. In particular, Zheng et al. \cite{ZhengR2016} found that a special APT-based cyber attack-defense model exhibits a global stability.

This paper focuses on estimating security of cyber networks under APTs. First, an APT-based cyber attack-defense process is modeled as an individual-level dynamical system. Second, the dynamic model is shown to exhibit the global stability. On this basis, a new security metric of cyber networks, which is known as the limit security, is defined as the limit expected fraction of compromised nodes in the networks. Next, the influence of different factors on the limit security is illuminated through theoretical analysis and computer simulation. This work helps understand the security of cyber networks under APTs.

The remaining materials are organized this way. Section 2 derives an APT-based cyber attack-defense model. Section 3 shows the global stability of the model and defines the limit security of cyber networks. The influence of different factors on the limit security is made clear in Sections 4 and 5. Finally, Section 6 closes this work.

\newtheorem{rk}{Remark}
\newproof{pf}{Proof}
\newtheorem{thm}{Theorem}
\newtheorem{lm}{Lemma}
\newtheorem{exm}{Example}
\newtheorem{cor}{Corollary}
\newtheorem{de}{Definition}
\newtheorem{cl}{Claim}
\newtheorem{pro}{Proposition}
\newtheorem{con}{Conjecture}
\newtheorem{expe}{Experiment}

\newproof{pfle2}{Proof of Lemma 2}
\newproof{pfth1}{Proof of Theorem 1}
\newproof{pfth4}{Proof of Theorem 4}
\newproof{pfcl1}{Proof of Claim 1}

\section{The modeling of APT-based cyber attack-defense processes}

For the purpose of evaluating the security of a cyber network under APTs, understanding the relevant cyber attack-defense process is requisite. This is the goal of this section.

\subsection{The cyber network}

Let $G = (V, E)$ denote the network interconnecting computers in a given cyber network, where $V = \{1, \cdots, N\}$, each node represents a computer in the cyber network, and there is an edge from node $i$ to node $j$ if and only if computer $i$ is allowed to deliver messages directly to computer $j$ through the network. Let $\mathbf{A} = \left(a_{ij}\right)_{N \times N}$ denote the adjacency matrix for $G$. Hereafter, $G$ is assumed to be strongly connected. 

In what follows, it is assumed that, at any time, every node in the cyber network is either \emph{secure} or \emph{compromised}, where all secure nodes are under the defender's control, and all compromised nodes are under the attacker's control. Let $X_i(t)$ = 0 and 1 denote that node $i$ is secure and compromised at time $t$, respectively. Then the state of the cyber network at time $t$ is represented by the vector 
\[
\mathbf{X}(t) = (X_1(t), X_2(t), \cdots, X_N(t)).
\]
Let $S_i(t)$ and $C_i(t)$ denote the probability of node $i$ being secure and compromised at time $t$, respectively.
\[
S_i(t) = \Pr\{X_i(t) = 0\}, \quad C_i(t) = \Pr\{X_i(t) = 1\}.
\] 
As $S_i(t)+C_i(t)\equiv 1$, the vector
$
\mathbf{C}(t) = (C_1(t), \cdots, C_N(t))^T
$
represents the expected state of the cyber network at time $t$.

\subsection{The attack and defense mechanisms}

The threat of an APT attack to the cyber network is twofold.

\begin{itemize}
	
	\item \emph{External attack}, which is conducted by the external attacker, with the intent of compromising the secure nodes in the network. The attack strength to secure node $i$ is $\alpha x_i$, where $\alpha$ stands for the technical level of external attack, $x_i$ stands for the resource per unit time used for attacking node $i$, $\sum_{i=1}^Nx_i > 0$.
	
	\item \emph{Internal infection}, which is caused by the compromised nodes in the network, with the intent of compromising the secure nodes in the network. The infection strength of compromised node $i$ to secure node $j$ is $\beta a_{ij}$, where $\beta$ stands for the technical level of internal infection. The combined infection strength to secure node $i$ at time $t$ is $f\left(\beta\sum_{j=1}^{N}a_{ji}1_{\{x_j(t) = 1\}}\right)$, where $1_A$ stands for the indicator function of event $A$, $f(0)=0$, $f(x) \leq x$ for all $x \geq 0$, $f$ is strictly increasing and concave, and $f$ is second continuously differentiable.
	
\end{itemize}

We refer to the vector $\mathbf{x} = (x_1, \cdots, x_N)$ as an attack scheme. The resource per unit time for the attack scheme $\mathbf{x}$ is $||\mathbf{x}||_1 = \sum_{i=1}^Nx_i$, where $||\cdot||_1$ stands for the 1-norm of vectors. .  

The defense of the cyber network against APTs is also twofold.

\begin{itemize}
	
	\item Prevention, which aims to prevent the secure nodes in the cyber network from being compromised. The prevention strength of secure node $i$ is $\delta y_i$, where $\delta$ stands for the technical level of prevention, $y_i > 0$ stands for the resource per unit time used for preventing node $i$.
	
	\item Recovery, which is intended to recover the compromised nodes in the cyber network. The recovery strength of compromised $i$ is $\gamma z_i$, where $\gamma$ stands for the technical level of recovery, $z_i > 0$ stands for the resource per unit time for recovering node $i$.
\end{itemize}

We refer to the vector $\mathbf{y} = (y_1, \cdots, y_N)$ as a prevention scheme, the vector $\mathbf{z} = (z_1, \cdots, z_N)$ as a recovery scheme, and the vector $(\mathbf{y}, \mathbf{z})$ as a defense scheme. The resources per unit time for the prevention scheme $\mathbf{y}$, the recovery scheme $\mathbf{y}$ and the defense scheme $(\mathbf{y}, \mathbf{z})$ are $||\mathbf{y}||_1$, $||\mathbf{z}||_1$ and $||\mathbf{y}||_1 + ||\mathbf{z}||_1$, respectively.

Let $\mathbf{w} = (w_1, \cdots, w_N)$ denote an attack scheme or a prevention scheme or a recovery scheme. In the subsequent study, the following two schemes will be used.

\begin{itemize}
	
\item The degree-first scheme: $w_i$ is linearly proportional to the out-degree of node $i$. Formally,
\[
  \mathbf{w} = ||\mathbf{w}||_1 \cdot \left(\frac{\sum_{j = 1}^N a_{1j}}{\sum_{i, j=1}^Na_{ij}}, \frac{\sum_{j = 1}^N a_{2j}}{\sum_{i, j=1}^Na_{ij}}, \cdots, \frac{\sum_{j = 1}^N a_{Nj}}{\sum_{i, j=1}^Na_{ij}}\right).
\]

\item The degree-last scheme: $w_i$ is inversely linearly proportional to the out-degree of node $i$. Formally,
\[
\mathbf{w} = ||\mathbf{w}||_1 \cdot \left(\frac{\frac{1}{\sum_{j = 1}^N a_{1j}}}{\sum_{i=1}^N \frac{1}{\sum_{j=1}^Na_{ij}}}, \frac{\frac{1}{\sum_{j = 1}^N a_{2j}}}{\sum_{i=1}^N \frac{1}{\sum_{j=1}^Na_{ij}}}, \cdots, \frac{\frac{1}{\sum_{j = 1}^N a_{Nj}}}{\sum_{i=1}^N \frac{1}{\sum_{j=1}^Na_{ij}}}\right).
\]

\item The uniform scheme: all $w_i$ are identical. Formally,
\[
\mathbf{w} = ||\mathbf{w}||_1 \cdot \left(\frac{1}{N}, \frac{1}{N}, \cdots, \frac{1}{N}\right).
\]

\end{itemize}

\subsection{The modeling of APT-based cyber attack-defense processes}

For the purpose of modeling APT-based cyber attack-defense processes, the following assumptions are made.

\begin{enumerate}

\item [(A$_1$)] Due to external attack, at any time secure node $i$ gets compromised at rate $\frac{\alpha x_i}{\delta y_i}$. This assumption is rational, because the rate is proportional to the attack strength and is inversely proportional to the prevention strength.

\item [(A$_2$)] Due to internal infection, at any time secure node $i$ gets compromised at rate $\frac{f\left(\beta\sum_{j=1}^{N}a_{ji}1_{\{x_j(t) = 1\}}\right)}{\delta y_i}$. This assumption is rational, because the rate is proportional to the combined infection strength and is inversely proportional to the prevention strength.

\item [(A$_3$)] Due to recovery, at any time compromised node $i$ becomes secure at rate $\gamma z_i$. This assumption is rational, because the rate is proportional to the recovery strength.

\end{enumerate}

Next, let us model the cyber attack-defense process. Let $\Delta t$ be a very small time interval. Following the above assumptions, we have that, for $i = 1, \cdots, N, t \geq 0$,
\[
\begin{split}
\Pr\{X_i(t+\Delta t) &= 1\mid X_i(t)=0\} = \frac{\Delta t}{\delta y_i}\left[\alpha x_i + f\left(\beta\sum_{j=1}^{N}a_{ji}C_j(t)\right)\right] + o(\Delta t), \\
\Pr\{X_i(t+\Delta t) & = 0\mid X_i(t)=1\} = \gamma z_i\Delta t + o(\Delta t).
\end{split}
\]
Invoking the total probability formula, rearranging the terms, dividing both sides by $\Delta t$, and letting $\Delta t\rightarrow 0$, we get a dynamic model as follows.
\begin{equation}
\frac{dC_i(t)}{dt}= \frac{\alpha x_i}{\delta y_i} - \left(\frac{\alpha x_i}{\delta y_i} + \gamma z_i\right)C_i(t) + \frac{1}{\delta y_i}[1 - C_i(t)]f\left(\beta\sum_{j=1}^{N}a_{ji}C_j(t)\right), \quad t \geq 0, i = 1, \cdots, N.
\end{equation}

\noindent We refer to the model as the generic secure-compromised-secure (GSCS) model, because the function $f$ meets a set of generic conditions. The diagram of state transitions of node $i$ under this model is given in Fig. 1. To a certain extent, the GSCS model accurately captures APT-based cyber attack-defense processes.

\begin{figure}[H]
	\centering
	\includegraphics[width=0.5\textwidth]{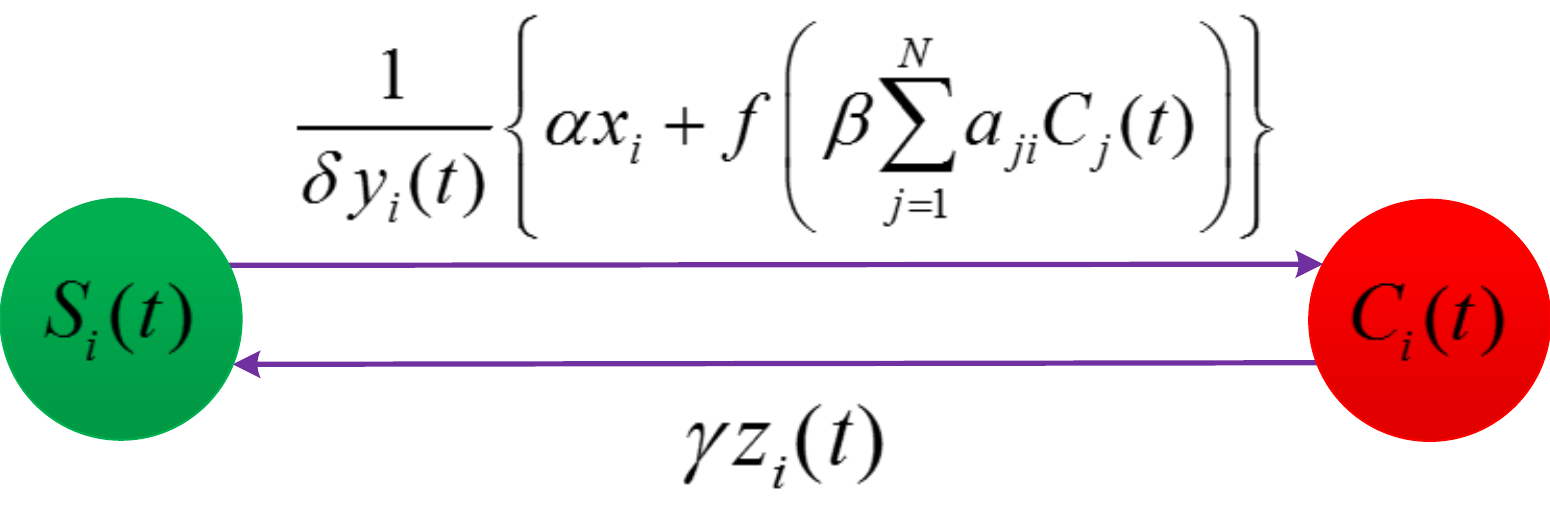}
	\caption{Diagram of state transitions of node $i$ under the GSCS model.}
\end{figure}

Let
\[
\Omega=\left\{(c_1,c_2,\cdots,c_N)^T\in \mathbb{R}_+^{N}\mid c_i\leq 1, i = 1, \cdots, N\right\}.
\]
It is trivial to show that $\mathbf{C}(t) \in \Omega$ for $t \geq 0$.

\section{Theoretical analysis of the GSCS model}

This section is dedicated to studying the dynamical properties of the GSCS model.

\subsection{Preliminaries}

For fundamental knowledge on differential dynamical systems, see Ref. \cite{Khalil2002}.

\begin{lm} (Chaplygin Lemma, see Theorem 31.4 in \cite{Szarski1965}) Consider a smooth $n$-dimensional system of differential equations 
	\[
	\frac{d\mathbf{x}(t)}{dt} = \mathbf{f}(\mathbf(\mathbf{x}(t)), \quad t \geq 0
	\]
	and the corresponding system of differential inequalities 
	\[
	\frac{d\mathbf{y}(t)}{dt} \geq \mathbf{f}(\mathbf(\mathbf{y}(t)), \quad t \geq 0
	\]
	with $\mathbf{x}(0) = \mathbf{y}(0)$. Suppose that for any $a_1, \cdots, a_n \geq 0$, there hold
	\[
	f_i(x_1+a_1, \cdots, x_{i-1}+a_{i-1}, x_i, x_{i+1} + a_{i+1}, \cdots, x_n + a_n) \geq f_i(x_1, \cdots, x_n), \quad i = 1, \cdots, n.
	\]
	Then $\mathbf{y}(t) \geq \mathbf{x}(t)$ for $t \geq 0$.
\end{lm}

For fundamental knowledge on fixed point theory, see Ref. \cite{Agarwal2001}.

\begin{lm} (Brouwer Fixed Point Theorem, see Theorem 4.10 in \cite{Agarwal2001}) Let $D$ be a nonempty, bounded, closed and convex subset of $\mathbb{R}^n$, and let $f: D \rightarrow D$ be a continuous function. Then $f$ has a fixed point.
\end{lm}

For fundamental knowledge on matrix theory, see Ref. \cite{Horn2013}. Let $diag(a_i)$ denote the diagonal matrix with diagnoal entries $a_1, a_2, \cdots, a_N$, and let $col(a_i)$ denote the column vector of components $a_1, a_2, \cdots, a_N$. This paper considers only real square matrices. For a matrix $\mathbf{A}$, let $s(\mathbf{A})$ denote the maximum real part of an eigenvalue of $\mathbf{A}$. $\mathbf{A}$ is \emph{Metzler} if its off-diagonal entries are all nonnegative.

\begin{lm}(Section 2.1 in \cite{Varga2000})
Let $\mathbf{A}$ be an irreducible Metzler matrix. Then the following claims hold.
\begin{enumerate}
	\item[(a)] If there is a positive vector $\mathbf{x}$ such that $\mathbf{A}\mathbf{x}<\lambda \mathbf{x}$, then $s(\mathbf{A})<\lambda $.
	\item[(b)] If there is a positive vector $\mathbf{x}$ such that $\mathbf{A}\mathbf{x}=\lambda \mathbf{x}$, then $s(\mathbf{A})=\lambda $.
	\item[(c)] If there is a positive vector $\mathbf{x}$ such that $\mathbf{A}\mathbf{x}>\lambda \mathbf{x}$, then $s(\mathbf{A})>\lambda $.
\end{enumerate}
\end{lm}

\subsection{A preliminary result}

For the GSCS model, let
\[
  \underline{C_i}=\frac{\alpha x_i}{\alpha x_i + \gamma\delta y_i z_i}, \quad \overline{C_i} = \frac{\alpha x_i+f(\beta\sum_{j=1}^{N}a_{ji})}{\alpha x_i + \gamma\delta y_i z_i + f(\beta\sum_{j=1}^{N}a_{ji})}, \quad i = 1, \cdots, N.
\]

The following lemma will be useful in the subsequent study.

\begin{lm}
Let $\mathbf{C}(t) = (C_1(t), C_2(t), \cdots, C_N(t))^T$ be a solution
to the SCS model. Then there are $t_0>0$ and $c>0$ such that
\[
  \min_{1 \leq i \leq N}C_i(t) \geq c, \quad t \geq t_0.
\]
\end{lm}
	
\begin{proof} Without loss of generality, assume $x_{i_0}>0$. It follows from the GSCS model that
\[
	\frac{dC_{i_0}(t)}{dt}\geq \frac{\alpha x_{i_0}}{\delta y_{i_0}}-\left(\frac{\alpha x_{i_0}}{\delta y_{i_0}} + \gamma z_{i_0}\right)C_{i_0}(t), \quad t\geq 0.
\]
Obviously, the comparison system 
\[
	\frac{du_{i_0}(t)}{dt}= \frac{\alpha x_{i_0}}{\delta y_{i_0}}-\left(\frac{\alpha x_{i_0}}{\delta y_{i_0}} + \gamma z_{i_0}\right)u_{i_0}(t), \quad t\geq 0,
\]
with $u_{i_0}(0) = C_{i_0}(0)$ admits $\underline{C_{i_0}} > 0$ as the globally stable equilibrium. By Lemma 1, we have 
\[
	C_{i_0}(t)\geq u_{i_0}(t), \quad t\geq 0. 
\]
So, 
\[
	\liminf_{t\rightarrow\infty}C_{i_0}(t)\geq \lim_{t \rightarrow \infty} u_{i_0}(t) = \underline{C_{i_0}}.
\]
Thus, for any $0 < \varepsilon < \underline{C_{i_0}}$, there is $t_1 > 0$ such that
\[
	C_{i_0}(t)\geq \underline{C_{i_0}} - \varepsilon, \quad t \geq t_1.
\]
As $G$ is strongly connected, there is $a_{i_0j_0} = 1$. Hence,
\[
	\frac{dC_{j_0}(t)}{dt}\geq \frac{1}{\delta y_{j_0}}f\left(\beta \left(\underline{C_{i_0}}-\varepsilon\right)\right)-\left[\frac{1}{\delta y_{j_0}}f\left(\beta \left(\underline{C_{i_0}}-\varepsilon\right)\right)+\gamma z_{j_0}\right]C_{j_0}(t), \quad t\geq t_1.
\]
Obviously, the comparison system
\[
	\frac{dv_{j_0}(t)}{dt} = \frac{1}{\delta y_{j_0}}f\left(\beta \left(\underline{C_{i_0}}-\varepsilon\right)\right)-\left[\frac{1}{\delta y_{j_0}}f\left(\beta \left(\underline{C_{i_0}}-\varepsilon\right)\right)+\gamma z_{j_0}\right]v_{j_0}(t), \quad t\geq t_1
\]
with $v_{j_0}(t_1)=C_{j_0}(t_1)$ admits $\frac{f\left(\beta \left(\underline{C_{i_0}}-\varepsilon\right)\right)}{f\left(\beta \left(\underline{C_{i_0}}-\varepsilon\right)\right)+\gamma\delta y_{j_0}z_{j_0}}$ as the globally stable equilibrium. By Lemma 1, we have
\[
	C_{j_0}(t) \geq v_{j_0}(t), \quad t \geq t_1.
\]
So,
\[
	\liminf_{t\rightarrow\infty}C_{j_0}(t)\geq \lim_{t \rightarrow \infty}v_{j_0}(t) = \frac{f\left(\beta \left(\underline{C_{i_0}}-\varepsilon\right)\right)}{f\left(\beta \left(\underline{C_{i_0}}-\varepsilon\right)\right)+\gamma\delta y_{j_0}z_{j_0}}.
\]
In view of the arbitrariness of $\varepsilon$, we get that 
\[
	\liminf_{t\rightarrow\infty}C_{j_0}(t)\geq \frac{f\left(\beta \underline{C_{i_0}}\right)}{f\left(\beta \underline{C_{i_0}}\right)+\gamma\delta y_{j_0}z_{j_0}} > 0. 
\]
The lemma follows by repeating the argument.
\end{proof}

\subsection{The equilibrium}

\begin{thm}
The GSCS model admits a unique equilibrium. Denote this equilibrium by $\mathbf{C}^*=(C_1^*,\cdots,C_N^*)^T$. Then $C_i^* > 0, \underline{C_i} \leq C_i^* \leq \overline{C_i}, 1 \leq i \leq N$.

\end{thm}

\begin{proof}
Let $K=\prod_{i=1}^N \left[\underline{C_i}, \overline{C_i}\right]$.
Define a continuous mapping $\mathbf{H} =(H_1,\cdots,H_N)^T: K \rightarrow [0, 1]^N$ as follows.
\[
H_i(\mathbf{w})=\frac{\alpha x_i+f(\beta\sum_{j=1}^{N}a_{ji}w_j)}{\alpha x_i + \gamma\delta y_iz_i + f(\beta\sum_{j=1}^{N}a_{ji}w_j)}, \quad\mathbf{w}=(w_1,\cdots,w_N)^T\in K.
\]
It is trivial to show that $\mathbf{C}$ is an equilibrium of the GSCS model if and only if $\mathbf{C}$ is a fixed point of $\mathbf{H}$. Furthermore, it is easy to show that $\mathbf{H}$ maps $K$ into itself. It follows from Lemma 2 that $\mathbf{H}$ has a fixed point, denoted $\mathbf{C}^*=(C_1^*,\cdots,C_N^*)^T$. This implies that $ \mathbf{C}^*$ is an equilibrium of the GSCS model, where $ \underline{C_i} \leq C_i^* \leq \overline{C_i}, 1 \leq i \leq N$. By Lemma 4, $C_i^* > 0, 1 \leq i \leq N$.

The remaining thing to do is to show that $\mathbf{C}^*$ is the unique fixed point of $\mathbf{H}$. On the contrary, suppose $\mathbf{H}$ has a fixed point other than $\mathbf{C}^*$. Denote this equilibrium by $\mathbf{C}^{**}=(C_1^{**},\cdots,C_N^{**})^T$. Let 
\[
\rho=\max_{1 \leq i \leq N}\frac{C_i^*}{C_i^{**}},\quad i_0=\arg\max_{1 \leq i \leq N} \frac{C_i^*}{C_i^{**}}.
\]

\noindent Without loss of generality, assume $\rho>1$. It follows that 
\[
\begin{split}
C_{i_0}^*&=H_{i_0}(\mathbf{C}^*)\leq H_{i_0}(\rho\mathbf{C}^{**})
=\frac{\alpha x_{i_0}+f(\rho\beta\sum_{j=1}^{N}a_{ji_0}C_j^{**})}{\gamma\delta y_{i_0}z_{i_0}+\alpha x_{i_0}+f(\rho\beta\sum_{j=1}^{N}a_{ji_0}C_j^{**})}
\\
&<\frac{\alpha x_{i_0}+f(\rho\beta\sum_{j=1}^{N}a_{ji_0}C_j^{**})}{\gamma\delta y_{i_0}z_{i_0}+\alpha x_{i_0}+f(\beta\sum_{j=1}^{N}a_{ji_0}C_j^{**})}\leq \frac{\alpha x_{i_0}+\rho f(\beta\sum_{j=1}^{N}a_{ji_0}C_j^{**})}{\gamma\delta y_{i_0}z_{i_0}+\alpha x_{i_0}+f(\beta\sum_{j=1}^{N}a_{ji_0}C_j^{**})}\\
&<\rho H_{i_0}(\mathbf{C}^{**})=\rho C_{i_0}^{**}.
\end{split}
\]

\noindent This contradicts the assumption that $C_{i_0}^*=\rho C_{i_0}^{**}$. Hence, $\mathbf{C}^*$ is the unique fixed point of $\mathbf{H}$. The proof is complete.
\end{proof}

\subsection{The stability of the equilibrium}

\begin{thm}
	The equilibrium $\mathbf{C}^*$ of the GSCS model is stable with recpect to $\Omega$. 
\end{thm}

\begin{proof}
Let $\mathbf{C}(t) = (C_1(t), C_2(t), \cdots, C_N(t))^T$ be a solution
to the GSCS model. By Lemma 4, there are $t_0>0$ and $c>0$ such that 
\[
  \min_{1 \leq i \leq N}C_i(t) \geq c, \quad t \geq t_0.
\]
Let 
	\[
	Z(\mathbf{C}(t))=\max_{1 \leq i \leq N} \frac{C_i(t)}{C_i^*}, \quad
	z(\mathbf{C}(t))=\min_{1 \leq i \leq N} \frac{C_i(t)}{C_i^*}, \quad t \geq t_0.
	\]
	Define a function $V$ as
	\[
	V(\mathbf{C}(t))=\max\{Z(\mathbf{C}(t))-1,0\}+\max\{1-z(\mathbf{C}(t)),0\}.
	\]
	
	\noindent It is easily verified that $V$ is positive definite with respect to $\mathbf{C}^*$, i.e., (a) $V(\mathbf{C}(t))\geq 0$, and (b) $V(\mathbf{C}(t))=0$ if and only if $\mathbf{C}(t)=\mathbf{C}^{*}$. Next , let us show that $D^+V(\mathbf{C}(t)) \leq 0, t \geq t_0$, where $D^+$ stands for the upper-right Dini derivative of $V$ along $\mathbf{C}(t)$. To this end, we need to show the following two claims.
	
	\emph{Claim 1:} $D^+Z(\mathbf{C}(t))\leq0$ if $Z(\mathbf{C}(t))\geq1$.
	Moreover, $D^+Z(\mathbf{C}(t))<0$ if $Z(\mathbf{C}(t))>1$.
	
	\emph{Claim 2:} $D_+z(\mathbf{C}(t))\geq0$ if $z(\mathbf{C}(t))\leq1$.
	Moreover, $D_+z(\mathbf{C}(t))> 0$ if  $z(\mathbf{C}(t))<1$. Here $D_+$ stands for the lower-right Dini derivative.
	
	\emph{Proof of Claim 1:} Choose $k_0$ such that  
	\[
	Z(\mathbf{C}(t))=\frac{C_{k_0}(t)}{C_{k_0}^*}, \quad D^+Z(\mathbf{C}(t))=\frac{C_{k_0}^{'}(t)}{C_{k_0}^*}.
	\] 
	Then, 
	\[
	\begin{split}
	\frac{C_{k_0}^{*}}{C_{k_0}(t)}C_{k_0}^{'}(t)
	&=\frac{\alpha x_{k_0}}{\delta y_{k_0}}\left(1-C_{k_0}(t)\right)\frac{C_{k_0}^{*}}{C_{k_0}(t)}+ \frac{1}{\delta y_{k_0}}\left(1-C_{k_0}(t)\right)\frac{C_{k_0}^{*}}{C_{k_0}(t)}f\left(\beta\sum_{j=1}^{N}a_{jk_0}C_j(t)\right)-\gamma z_{k_0} C_{k_0}^{*}\\
	&\leq
	\frac{\alpha x_{k_0}}{\delta y_{k_0}}\left(1-C_{k_0}^{*}\right)+ \frac{1}{\delta y_{k_0}}\left(1-C_{k_0}^{*}\right)\frac{C_{k_0}^{*}}{C_{k_0}(t)}f\left(\beta\sum_{j=1}^{N}a_{jk_0}C_j(t)\right)-\gamma z_{k_0} C_{k_0}^{*} \\
	&\leq 
	\frac{\alpha x_{k_0}}{\delta y_{k_0}}\left(1-C_{k_0}^{*}\right)+ \frac{1}{\delta y_{k_0}}\left(1-C_{k_0}^{*}\right)f\left(\beta\frac{C_{k_0}^{*}}{C_{k_0}(t)}\sum_{j=1}^{N}a_{jk_0}C_j(t)\right)-\gamma z_{k_0} C_{k_0}^{*}\\
	& \leq \frac{\alpha x_{k_0}}{\delta y_{k_0}}\left(1-C_{k_0}^{*}\right) + \frac{1}{\delta y_{k_0}}\left(1-C_{k_0}^{*}\right)f\left(\beta\sum_{j=1}^{N}a_{jk_0}C_j^{*}\right)-\gamma z_{k_0} C_{k_0}^{*}=0,
	\end{split}
	\]
	
	\noindent where the second inequality follows from the concavity of $f$, and the third inequality follows from the monotonicity of $f$. This implies $D^+Z(\mathbf{C}(t))\leq0$. As the first inequality is strict if $Z(\mathbf{C}(t))>1$, we get that $D^+Z(\mathbf{C}(t))<0$ if $Z(\mathbf{C}(t))>1$. Claim 1 is proven.
	
	The argument for Claim 2 is analogous to that for Claim 1 and hence is omitted. Next, consider three possibilities.
	
	Case 1: $Z(\mathbf{C}(t)) < 1$. Then $z(\mathbf{C}(t)) < 1$, $V(\mathbf{C}(t)) = 1 - z(\mathbf{C}(t))$. Hence, 
	$D^+V(\mathbf{C}(t)) = -D_+z(\mathbf{C}(t)) < 0$.
	
	Case 2: $z(\mathbf{C}(t)) > 1$. Then $Z(\mathbf{C}(t)) > 1$, $V(\mathbf{C}(t)) = Z(\mathbf{C}(t)) - 1$. Hence, 
	$D^+V(\mathbf{C}(t)) = D^+Z(\mathbf{C}(t)) < 0$.
	
	Case 3: $Z(\mathbf{C}(t)) \geq 1$, $z(\mathbf{C}(t)) \leq 1$. Then $V(\mathbf{C}(t)) = Z(\mathbf{C}(t)) - z(\mathbf{C}(t))$, 
	$D^+V(\mathbf{C}(t)) = D^+Z(\mathbf{C}(t)) - D_+z(\mathbf{C}(t)) \leq 0$.
	Moreover, the equality holds if and only if $\mathbf{C}(t) = \mathbf{C}^*$.
	
	The theorem follows from the LaSalle Invariance Principle.
\end{proof}

Let $C(t)$ denote the expected fraction of compromised nodes in the cyber network at time $t$, $C^*$ the expected fraction of compromised nodes in the cyber network when the expected network state is $\mathbf{C}^*$.
\begin{equation}
  C(t) = \frac{1}{N}\sum_{i=1}^NC_i(t), \quad C^* = \frac{1}{N}\sum_{i=1}^NC_i^*.
\end{equation}
The following result is a corollary of Theorem 2.

\begin{cor} 
	Consider the GSCS model (1). Then $C(t) \rightarrow C^*$ as $t \rightarrow \infty$.
\end{cor}

Obviously, $C^*$ is dependent upon the four technical levels, the interconnection network, and the attack and defense schemes. We refer to the four technical levels and the interconnection network as \emph{parameters}, because they are almost fixed. We refer to the attack and defense schemes as \emph{independent variables}, because the attack scheme is flexibly choosable by the attacker, and the defense scheme is flexibly choosable by the defender. Formally, 
\[
C^* = C^*(\mathbf{x}, \mathbf{y}, \mathbf{z}; \alpha, \beta, \delta, \gamma, G).
\]

\subsection{The limit security of cyber networks}

In practice, $C^*$ can be estimated simply through sampling and averaging. This method for estimating $C^*$ is valuable, because it does not require the defender to know the attack and infection tecnical levels as well as the attack scheme. Therefore, $C^*$ can be used to evaluate the security of the cyber network. Below let us define a security metric of cyber networks under APTs.

\begin{de}
Given the four technical levels, the interconnected network, the attack scheme and the defense scheme, the \emph{limit security} of the cyber network is defined as
\begin{equation}
  S_L = 1 - C^*, 
\end{equation}
\end{de}

This security metric of cyber networks is rational, because the higher the limit security, the securer the cyber network would be. The limit security is dependent upon the four technical levels, the interconnection network, the attack scheme and the defense scheme. Formally,
\[
C^* = C^*(\mathbf{x}, \mathbf{y}, \mathbf{z}; \alpha, \beta, \delta, \gamma, G).
\]

\section{The influence of some factors on the limit security of a cyber network}

In this section, we theoretically investigate the influence of some factors, including the technical levels, the attack and defense resources per unit time per node, and the addition of new edges to the interconnection network, on the limit security of a cyber network. For this purpose, define an irreducible Metzler matrix as follows.
\[
\mathbf{M} = diag \left(\beta(1-C_i^*)f^{'}\left(\beta \sum_{j=1}^{N}a_{ji}C_j^*\right)\right)\mathbf{A}^T - diag\left(\alpha x_i + \gamma \delta y_i z_i + f\left(\beta \sum_{j=1}^N a_{ji}C_j^*\right)\right).
\]

\begin{lm}
	$\mathbf{M}$ is invertible, and $\mathbf{M}^{-1}$ is negative.
\end{lm}

\begin{proof} 
	As $f$ is concave, we have 
	\[
	f^{'}\left(\beta \sum_{j=1}^{N}a_{ji}C_j^*\right) \leq \frac{f\left(\beta \sum_{j=1}^{N}a_{ji}C_j^*\right)}{\beta\sum_{j=1}^Na_{ji}c_j^*}.
	\]
	So,
	\[
	\begin{split}
	\mathbf{M}\mathbf{C}^* &= diag \left(\beta(1-C_i^*)f^{'}\left(\beta \sum_{j=1}^{N}a_{ji}C_j^*\right)\right)\mathbf{A}^T\mathbf{C}^* - diag\left(\alpha x_i + \gamma \delta y_i z_i + f\left(\beta \sum_{j=1}^N a_{ji}C_j^*\right)\right)\mathbf{C}^*\\
	& \leq diag \left(\beta(1-C_i^*)\frac{f\left(\beta \sum_{j=1}^{N}a_{ji}C_j^*\right)}{\beta\sum_{j=1}^Na_{ji}c_j^*}\right)\mathbf{A}^T\mathbf{C}^* - diag\left(\alpha x_i + \gamma \delta y_i z_i + f\left(\beta \sum_{j=1}^N a_{ji}C_j^*\right)\right)\mathbf{C}^*\\
	&= - col\left(\alpha x_i + f\left(\beta \sum_{j=1}^Na_{ji}C_j^*\right)C_i^*\right) < \mathbf{0}.
	\end{split}
	\] 
	It follows from Lemma 3(a) that $s(\mathbf{M})<0$. This implies that $\mathbf{M}$ is invertible. As $\mathbf{M}$ is Metzler, irreducible and Hurwitz, $\mathbf{M}^{-1}$ is negative \cite{Narendra2010}.
\end{proof}

\subsection{The influence of the four technical levels}

\begin{thm} 
For the GSCS model (1), we have $\frac{\partial \mathbf{C}^*}{\partial \alpha} > \mathbf{0}$, $\frac{\partial \mathbf{C}^*}{\partial \beta} > \mathbf{0}$, $\frac{\partial \mathbf{C}^*}{\partial \gamma} < \mathbf{0}$, $\frac{\partial \mathbf{C}^*}{\partial \delta} < \mathbf{0}$.
\end{thm}

\begin{proof}
	We prove only $\frac{\partial \mathbf{C}^*}{\partial \beta} > \mathbf{0}$, because the arguments for the remaining claims are similar. As $\mathbf{C}^*$ is the equilibrium for the GSCS model, we have
	\[
	F_i(\beta; C_1^*, C_2^*, \cdots, C_N^*) = \alpha x_i - (\alpha x_i + \gamma\delta y_i z_i) C_i^* + (1-C_i^*)f\left(\beta\sum_{j=1}^Na_{ji}C_j^*\right) = 0, \quad 1\leq i\leq N.
	\]
	Differentiating on both sides with respect to $\beta$, we get
	\[
	\frac{\partial F_i}{\partial \beta} + \frac{\partial F_i}{\partial C_1^*} \cdot \frac{\partial C_1^*}{\partial \beta} + \cdots + \frac{\partial F_i}{\partial C_N^*} \cdot \frac{\partial C_N^*}{\partial \beta} = 0,\quad 1\leq i\leq N.
	\]
	Calculations show that
	\[
	\mathbf{M}\frac{\partial \mathbf{C}^*}{\partial \beta} = - diag \left((1-C_i^*)f^{'}(\beta \sum_{j=1}^{N}a_{ji}C_j^*)\right)\mathbf{A}^T\mathbf{C}^*.
	\]
	By Lemma 5, we have
	\[
	\frac{\partial \mathbf{C}^*}{\partial \beta} = - \mathbf{M}^{-1} \cdot diag \left((1-C_i^*)f^{'}(\beta \sum_{j=1}^{N}a_{ji}C_j^*)\right)\mathbf{A}^T\mathbf{C}^*.
	\]
	where $\mathbf{M}^{-1}$ is negative. As $G$ is strongly connected, $\mathbf{A}^T\mathbf{C}^*$ is positive. Hence, $\frac{\partial \mathbf{C}^*}{\partial \beta} > \mathbf{0}$.
\end{proof}

As a corollary of this theorem, the influence of the four technical levels on the limit security of a cyber network is shown as follows.

\begin{cor} 
	For the GSCS model (1), we have $\frac{\partial S_L}{\partial \alpha} < 0$, $\frac{\partial S_L}{\partial \beta} < 0$, $\frac{\partial S_L}{\partial \gamma} > 0$, $\frac{\partial S_L}{\partial \delta} > 0$.
\end{cor}

This corollary manifests that the limit security of a cyber network goes up with the prevention and recovery technical levels, and comes down with the attack and infection technical levels. These results accord with our intuition. Hence, the defender must try his best to enhance the prevention and recovery technical levels.

\subsection{The influence of the attack and defense resources per unit time per node}

\begin{thm}
For the GSCS model (1), we have $\frac{\partial \mathbf{C}^*}{\partial x_i} > \mathbf{0}$, $\frac{\partial \mathbf{C}^*}{\partial y_i} < \mathbf{0}$, $\frac{\partial \mathbf{C}^*}{\partial z_i} < \mathbf{0}$, $1\leq i\leq N$.
\end{thm}

The argument for the theorem is analogous to that for the previous theorem. As a corollary of this theorem, the influence of the attack and defense resources per unit time per node on the limit security of a cyber network is shown as follows.

\begin{cor} 
	For the GSCS model (1), we have $\frac{\partial S_L}{\partial x_i} < 0$, $\frac{\partial S_L}{\partial y_i} > 0$, $\frac{\partial S_L}{\partial z_i} > 0$, $1\leq i\leq N$.
\end{cor}

This corollary demonstrates that the limit security of a cyber network rises with the resource per unit time used for preventing or recovering a node, and falls with the resource per unit time used for attacking a node. Again, these results are consistent with our intuition. As a consequence, the defender is suggested to configure more defense resource.

\subsection{The influence of the addition of new edges to the interconnection network}

\begin{thm}
For the GSCS model, we have $\frac{\partial \mathbf{C}^*}{\partial a_{ij}} > \mathbf{0}$, $1\leq i, j\leq N$, $i \neq j$.
\end{thm}

The argument for the theorem is analogous to that for Theorem 3. As a corollary of this theorem, the addition of new edges to the interconnection network on the limit security of a cyber network is shown as follows.

\begin{cor} 
	For the GSCS model (1), we have $\frac{\partial S_L}{\partial a_{ij}} < 0$, $1\leq i, j\leq N$, $i \neq j$.
\end{cor}

This corollary manifests that the limit security of a cyber network declines with the addition of new edges to the interconnection network. Hence, a well-connected cyber network is more vulnerable to APTs. Therefore, the defender is suggested to limit the number of connections in the interconnection network.

\section{The influence of two other factors on the limit security of a cyber network}

In this section, we experimentally examine the influence of two factors, the ratio of the prevention resource to the recovery resource, and the defense resource per unit time with given ratio of the attack resource to the defense resource, on the limit security of a cyber network. In the following experiments, the generic function in the GSCS model is set to be $f(x) = \frac{x}{1 + x}$, and the interconnection network takes value from a set of six non-isomorphic trees shown in Fig. 2.

\begin{figure}[H]
	\subfigure[$G_1$]{\includegraphics[width=0.2\textwidth]{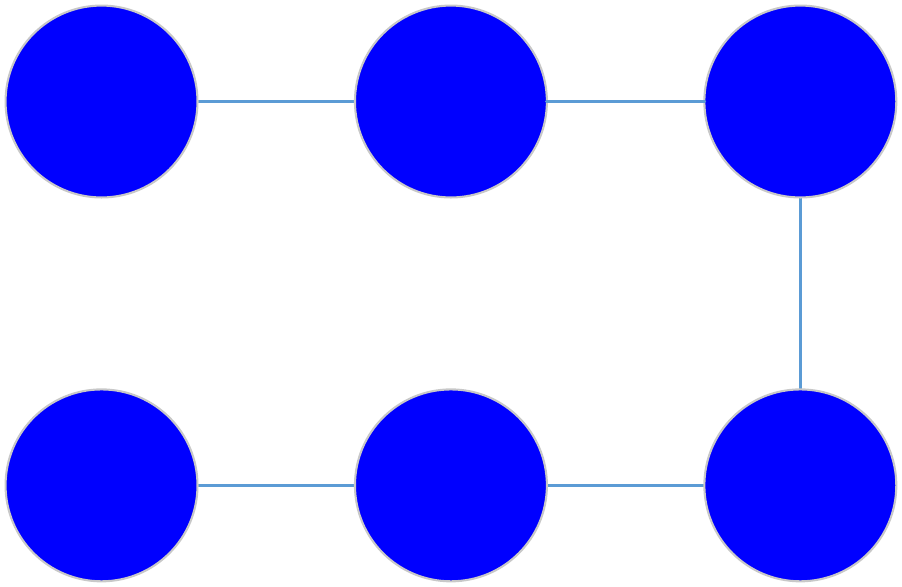}\label{a}}
	\hspace{5ex}
	\subfigure[$G_2$]{\includegraphics[width=0.33\textwidth]{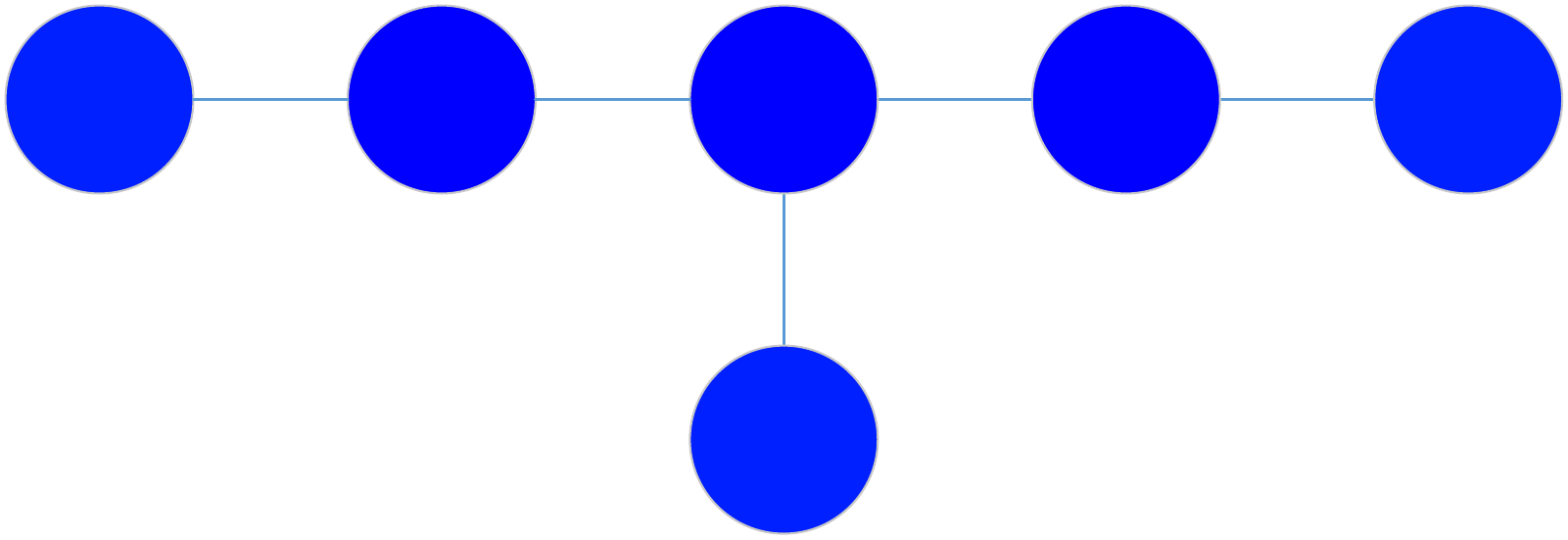}\label{a}}
	\hspace{5ex}
	\subfigure[$G_3$]{\includegraphics[width=0.33\textwidth]{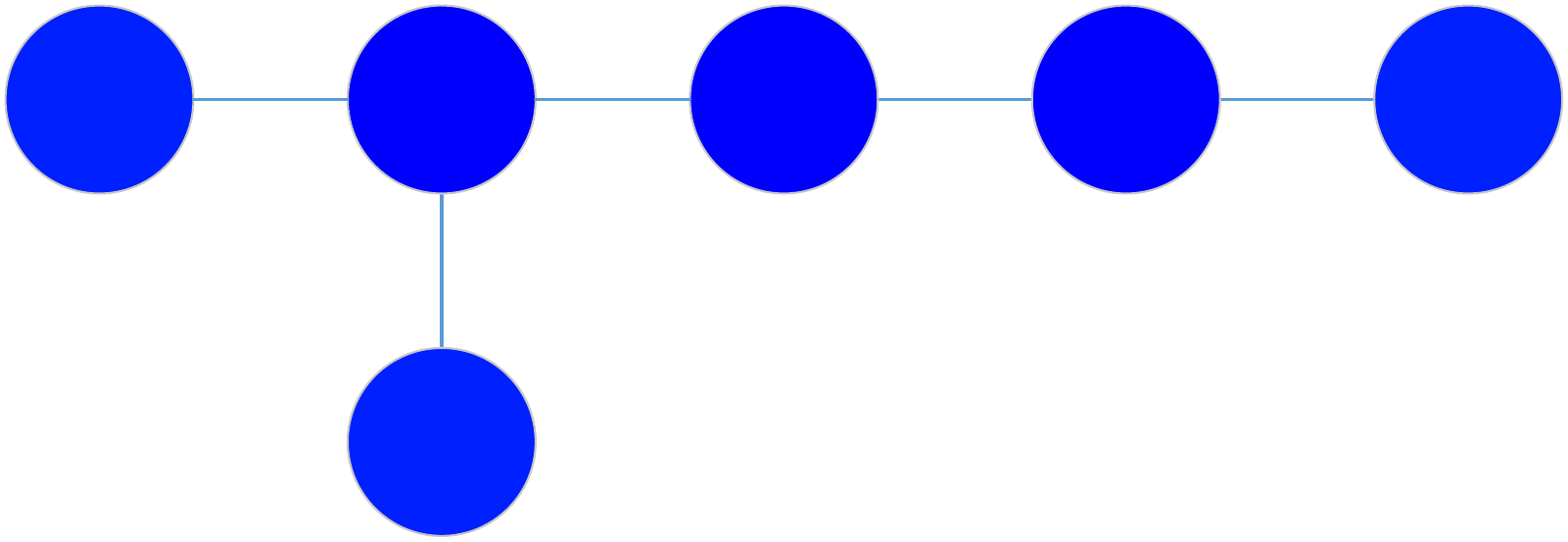}\label{a}}
	\hspace{2ex}
	\subfigure[$G_4$]{\includegraphics[width=0.3\textwidth]{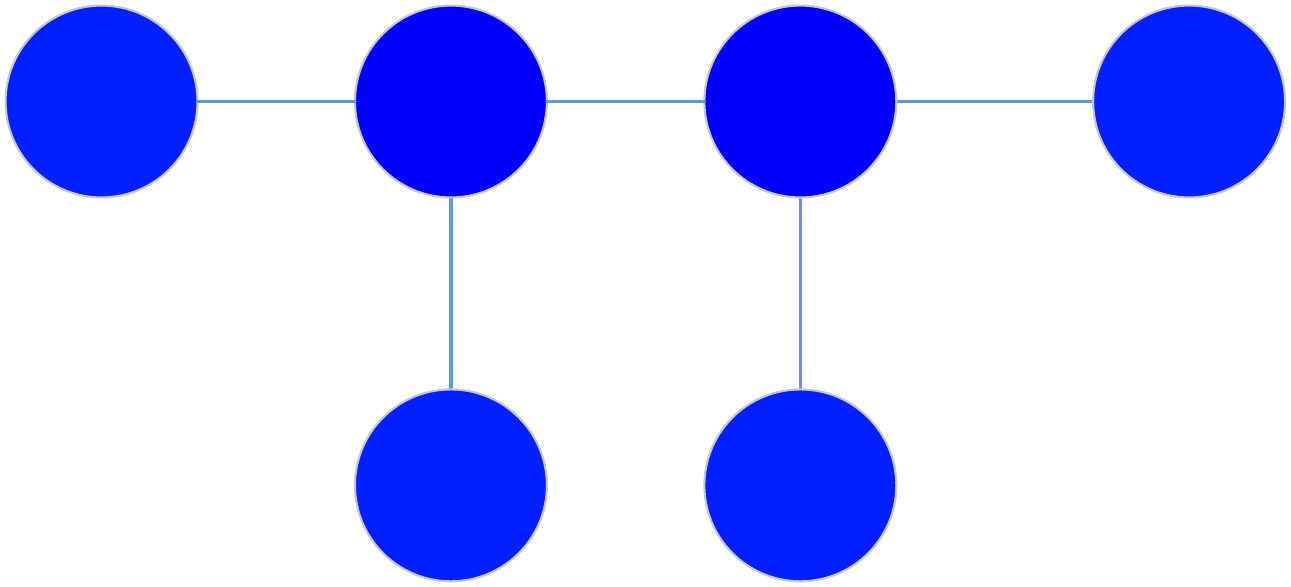}\label{a}}
	\hspace{5ex}
	\subfigure[$G_5$]{\includegraphics[width=0.3\textwidth]{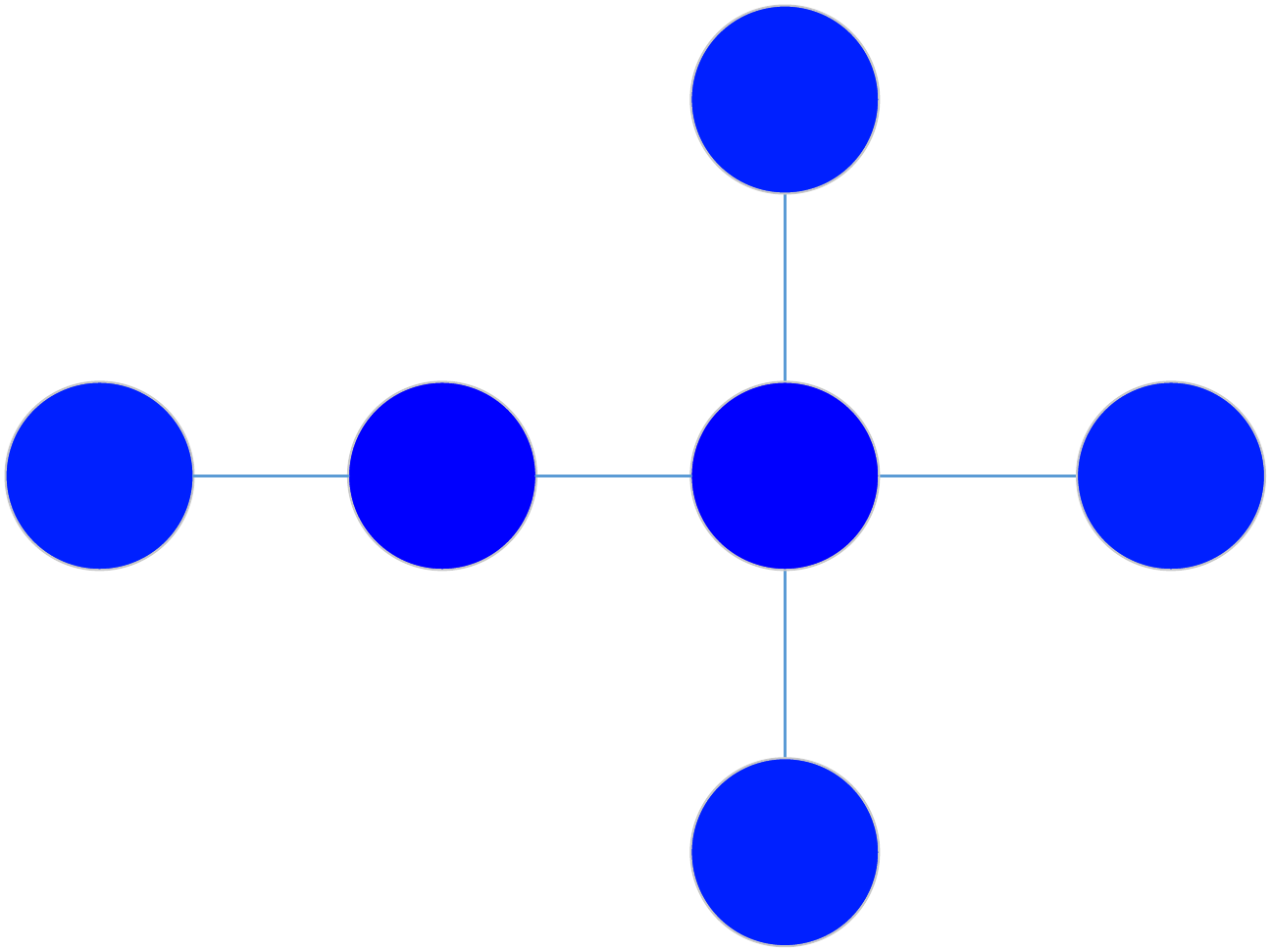}\label{a}}
	\hspace{10ex}	
	\subfigure[$G_6$]{\includegraphics[width=0.2\textwidth]{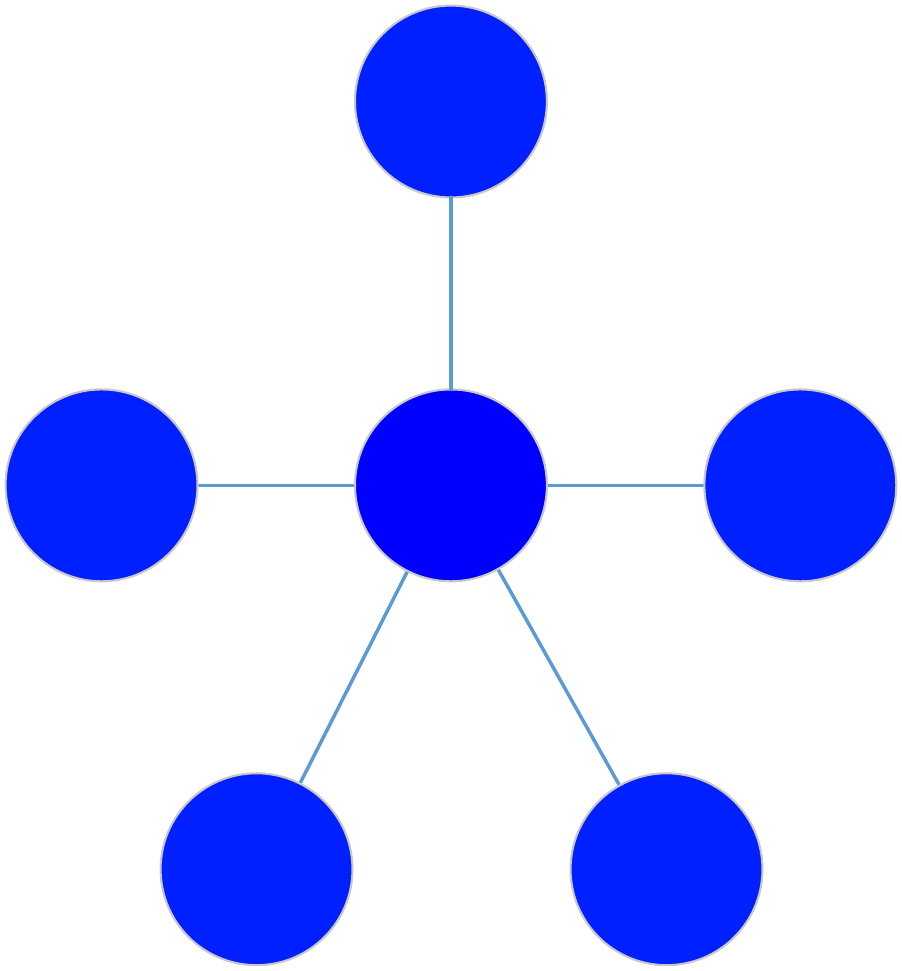}\label{a}}
	\caption{Six non-isomorphic trees with six nodes and five edges.}
	\vspace{2ex}
\end{figure}

\subsection{The influence of the ratio of the prevention resource to the recovery resource}

For a GSCS model, the ratio of the prevention resource to the recovery resource is 
\[
  r_{PR} = \frac{||\mathbf{y}||_1}{||\mathbf{z}||_1}. 
\]
We examine the influence of $r_{PR}$ on the limit security of a cyber network through simulation experiments.

\begin{expe}
	Consider 504 GSCS models, where $\alpha = 0.05$, $\beta = 0.01$, $\gamma = 1$, $\delta = 1$, $G$ varies from $G_1$ to $G_6$, $||\mathbf{x}||_1 = 1$, $||\mathbf{y}||_1 = \frac{r}{1 + r}$, $||\mathbf{z}||_1 = \frac{1}{1 + r}$, $r \in \{\frac{1}{4}, \frac{1}{3}, \frac{1}{2}, 1, 2, 3, 4\}$, with (a) uniform $\mathbf{x}$, $\mathbf{y}$ and $\mathbf{z}$; (b) uniform $\mathbf{x}$ and $\mathbf{y}$, degree-first $\mathbf{z}$; (c) uniform $\mathbf{x}$ and $\mathbf{z}$, degree-first $\mathbf{y}$; (d) uniform $\mathbf{x}$, degree-first $\mathbf{y}$ and $\mathbf{z}$; (e) degree-first $\mathbf{x}$, uniform $\mathbf{y}$ and $\mathbf{z}$; (f) degree-first $\mathbf{x}$ and $\mathbf{z}$, uniform $\mathbf{y}$; (g) degree-first $\mathbf{x}$ and $\mathbf{y}$, uniform $\mathbf{z}$; (h) degree-first $\mathbf{x}$, $\mathbf{y}$ and $\mathbf{z}$; (i) degree-last $\mathbf{x}$, uniform $\mathbf{y}$ and $\mathbf{z}$; (j) degree-last $\mathbf{x}$, uniform $\mathbf{y}$, degree-first $\mathbf{z}$; (k) degree-last $\mathbf{x}$, degree-first $\mathbf{y}$, uniform $\mathbf{z}$; (l) degree-last $\mathbf{x}$, degree-first $\mathbf{y}$ and $\mathbf{z}$. For each of the GSCS model, the limit security of the cyber network is shown  shown in Fig. 3. It can be seen that, with the increase of $r_{PR}$, the limit security of a cyber network goes up first but then it goes down. Moreover, the limit security attains the maximum in the proximity of $r_{PR} = 1$.
\end{expe}

Many similar experiments exhibit qualitatively similar phenomena. It is concluded that, with the increase of the ratio of the prevention resource to the recovery resource, the limit security of a cyber network goes up first but then it goes down. Moreover, the limit security attains the maximum when the prevention resource is close to the recovery resource. Hence, the defender is suggested to distribute the total defense resource equally to prevention and recovery. 

\begin{figure}[H]
	\centering
	\subfigure[]{\includegraphics[width=0.23\textwidth,height=3.2cm]{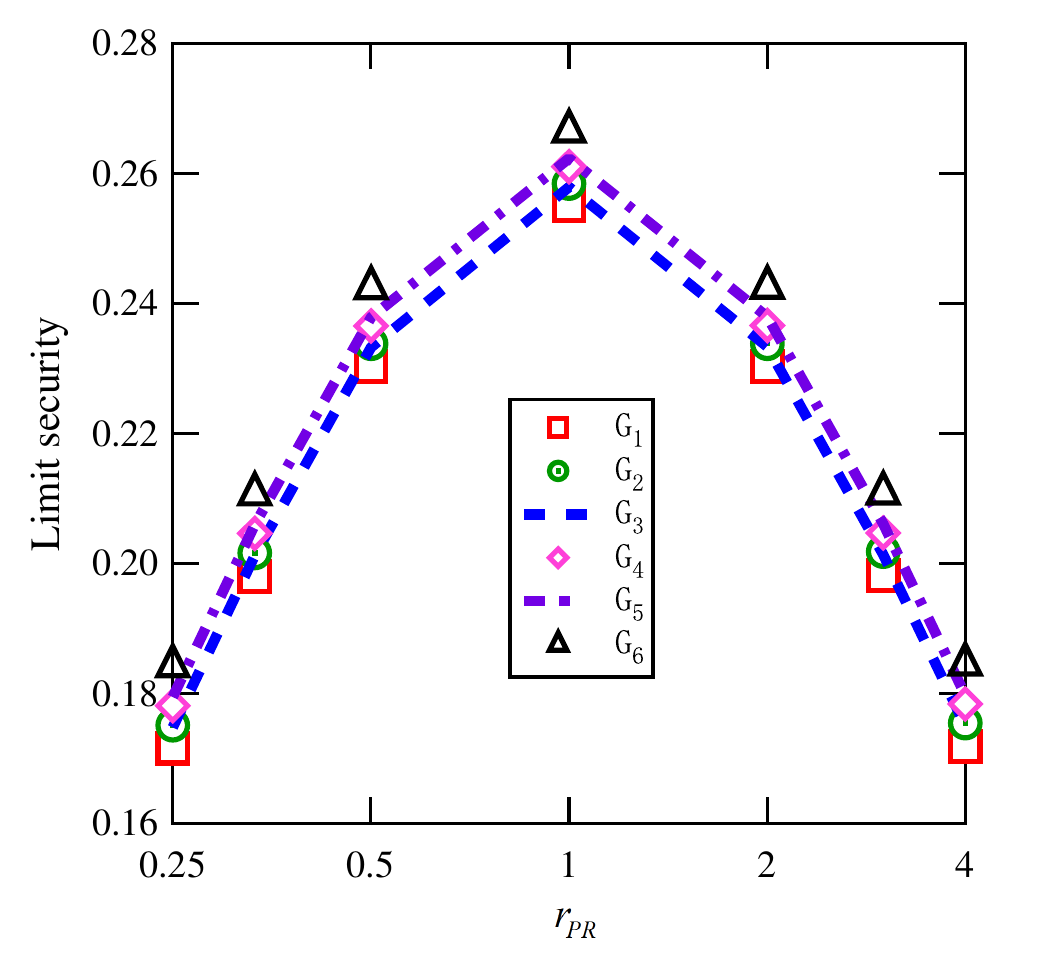}\label{a}}
	\subfigure[]{\includegraphics[width=0.23\textwidth,height=3.2cm]{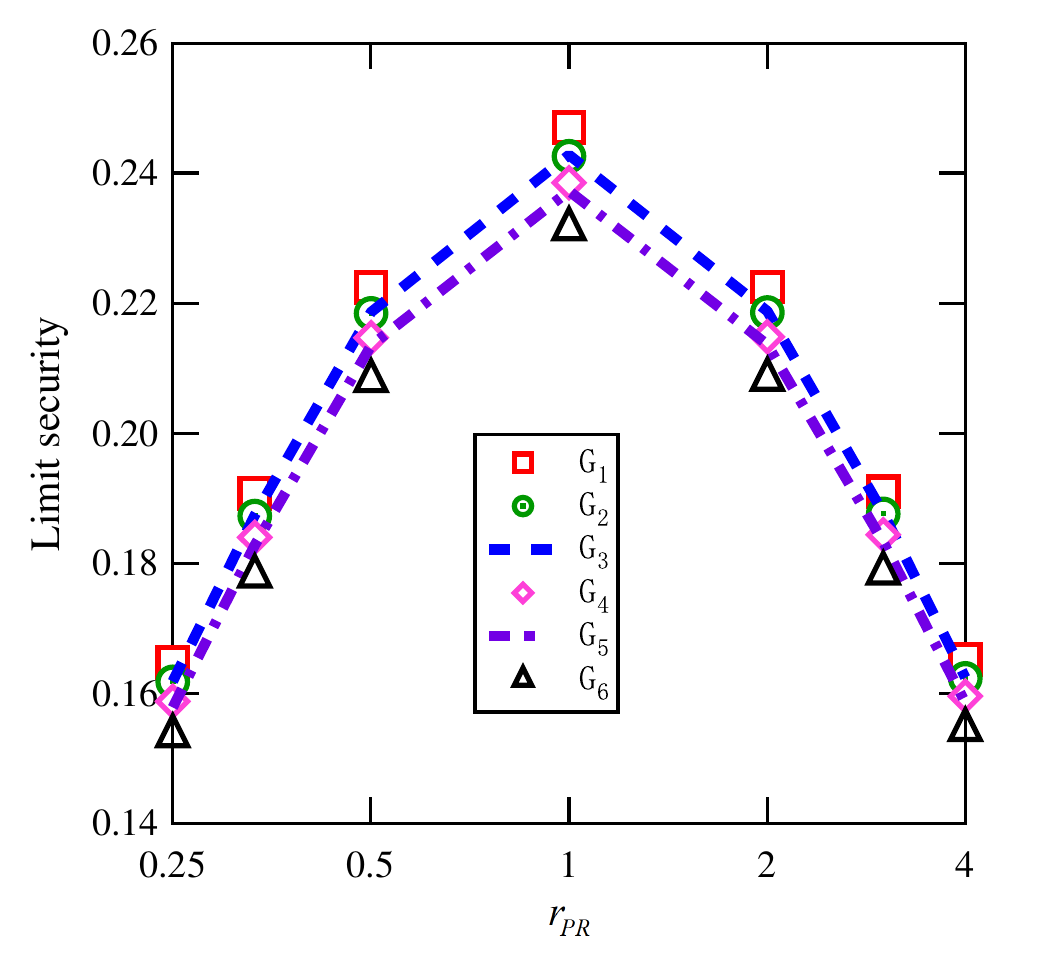}\label{a}}
	\subfigure[]{\includegraphics[width=0.23\textwidth,height=3.2cm]{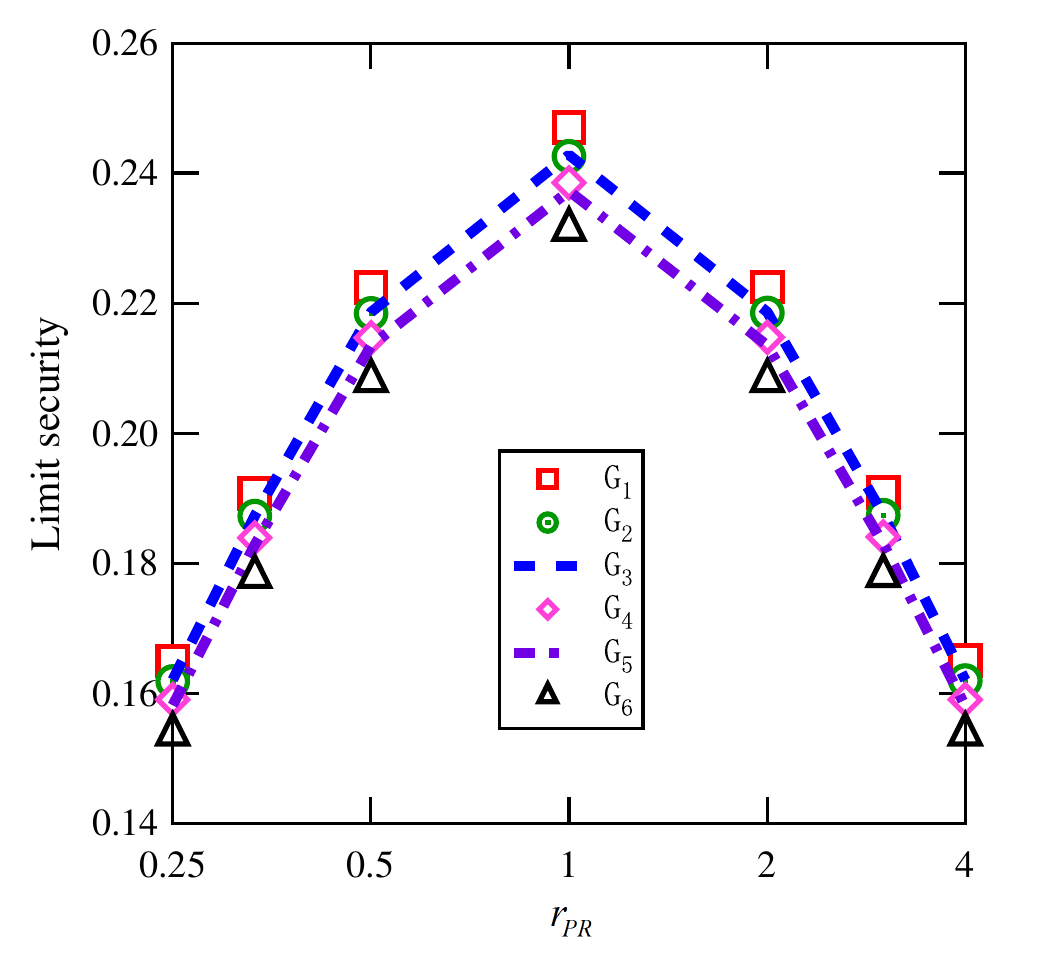}\label{a}}
	\subfigure[]{\includegraphics[width=0.23\textwidth,height=3.2cm]{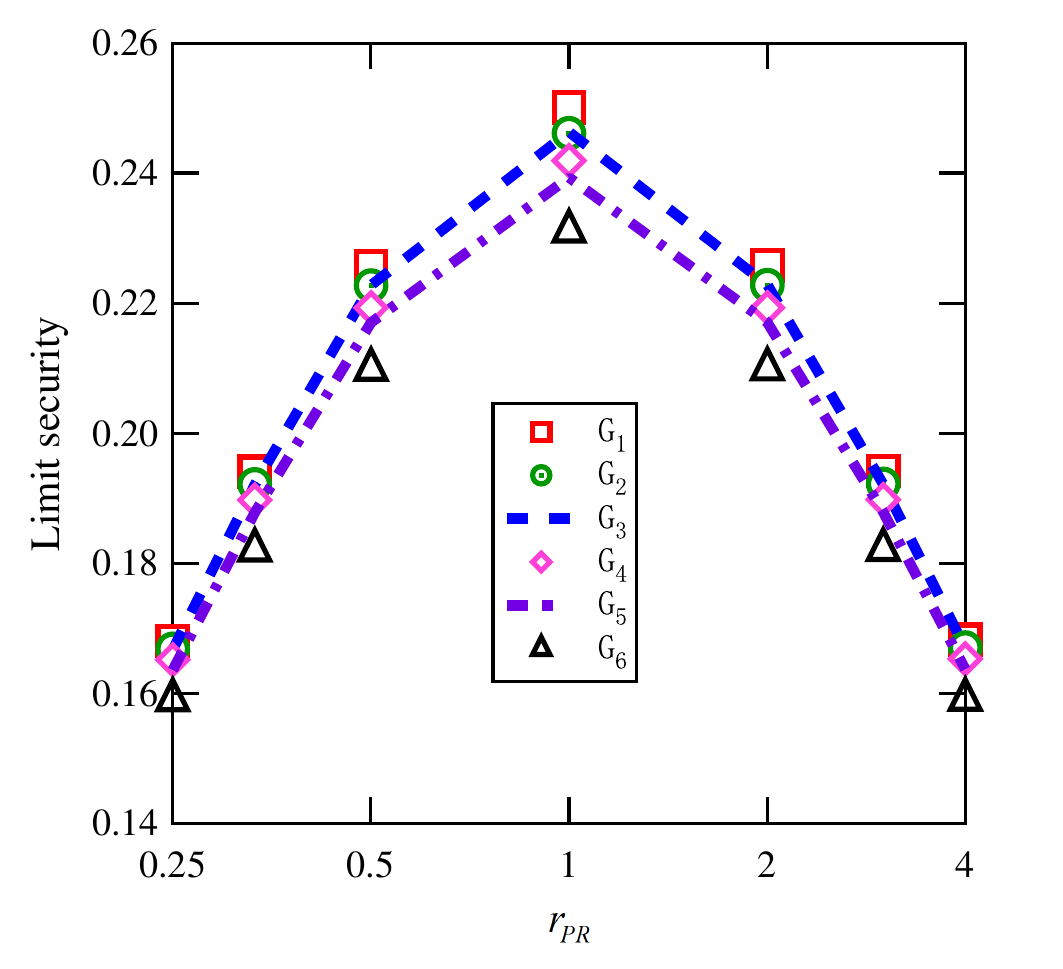}\label{a}}\\
	\subfigure[]{\includegraphics[width=0.23\textwidth,height=3.2cm]{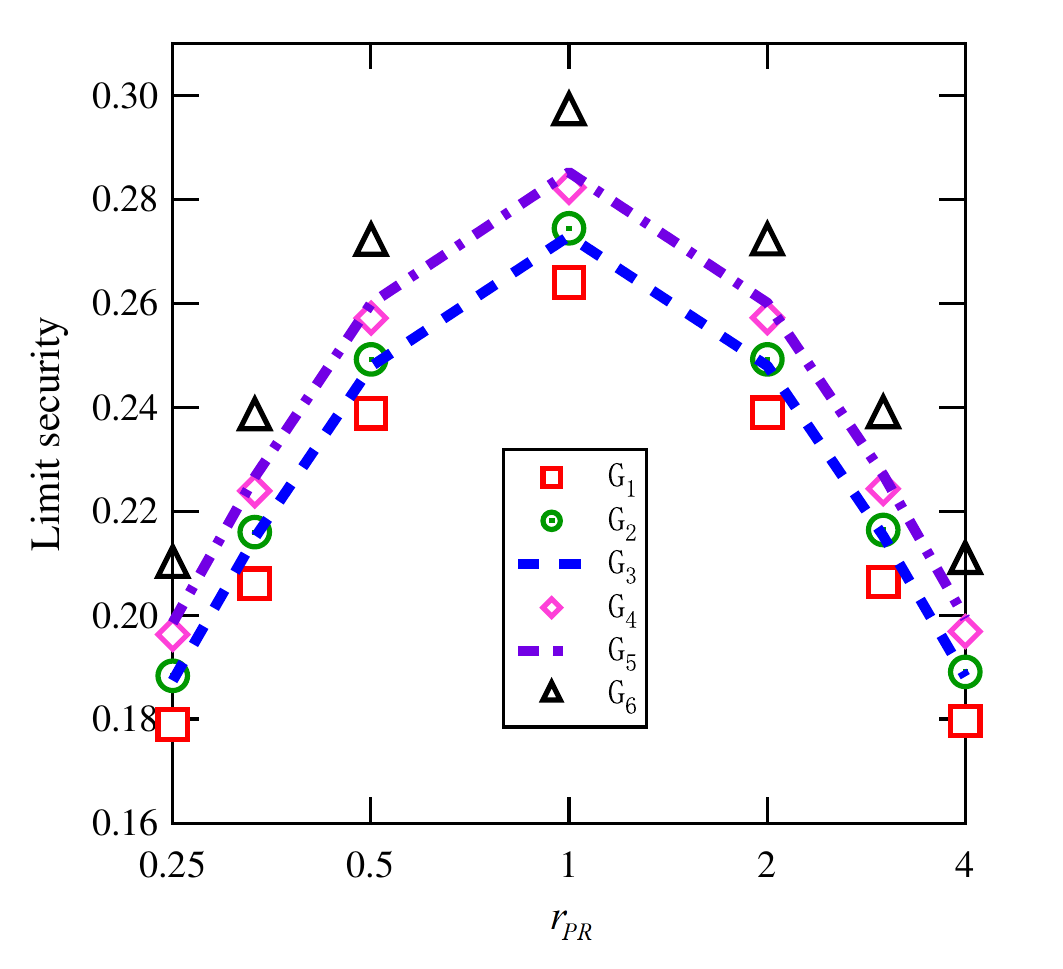}\label{a}}
	\subfigure[]{\includegraphics[width=0.23\textwidth,height=3.2cm]{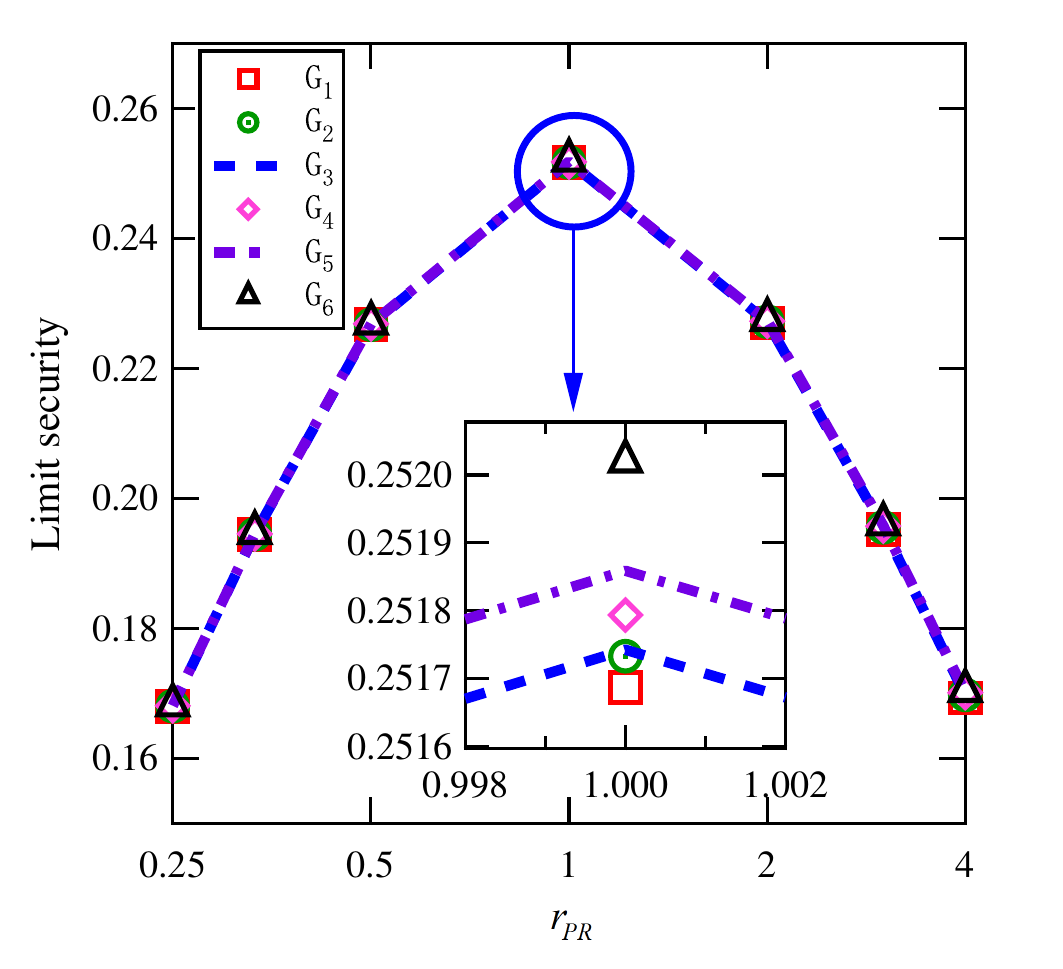}\label{a}}
	\subfigure[]{\includegraphics[width=0.23\textwidth,height=3.2cm]{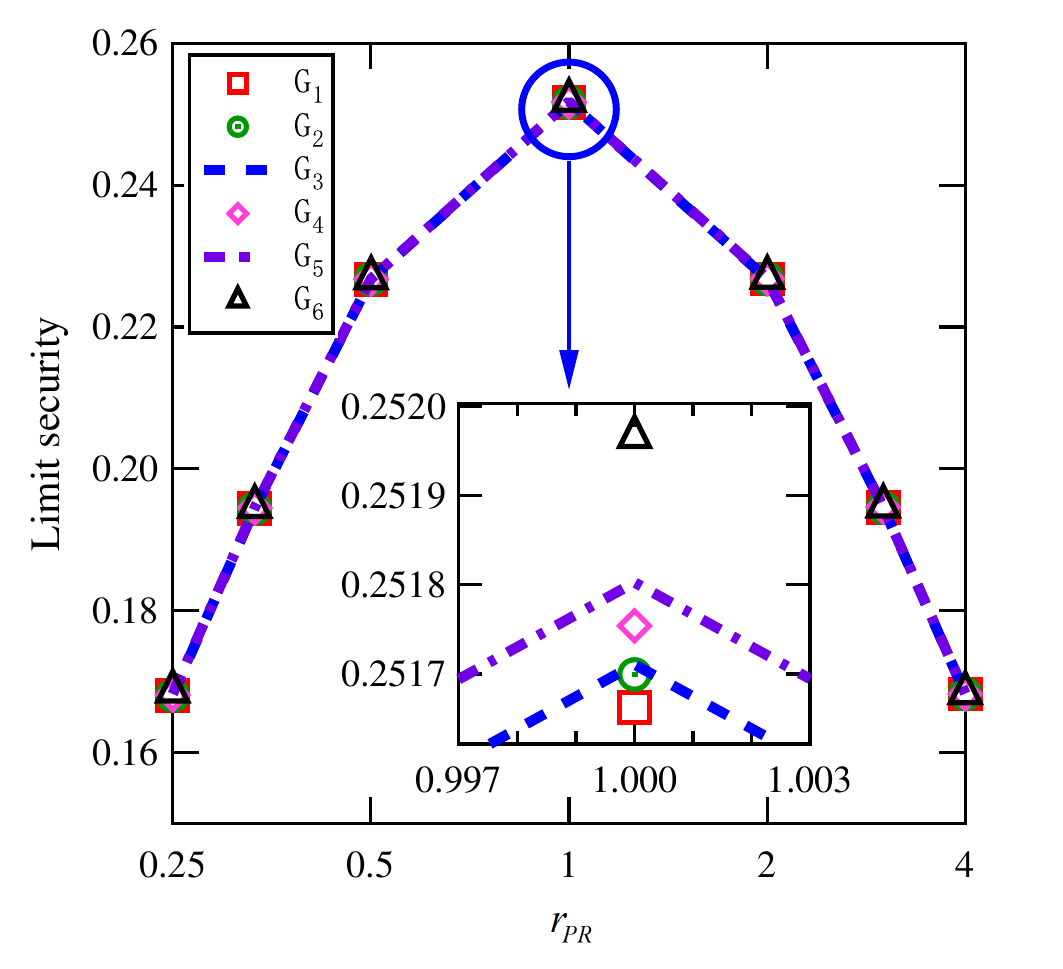}\label{a}}
	\subfigure[]{\includegraphics[width=0.23\textwidth,height=3.2cm]{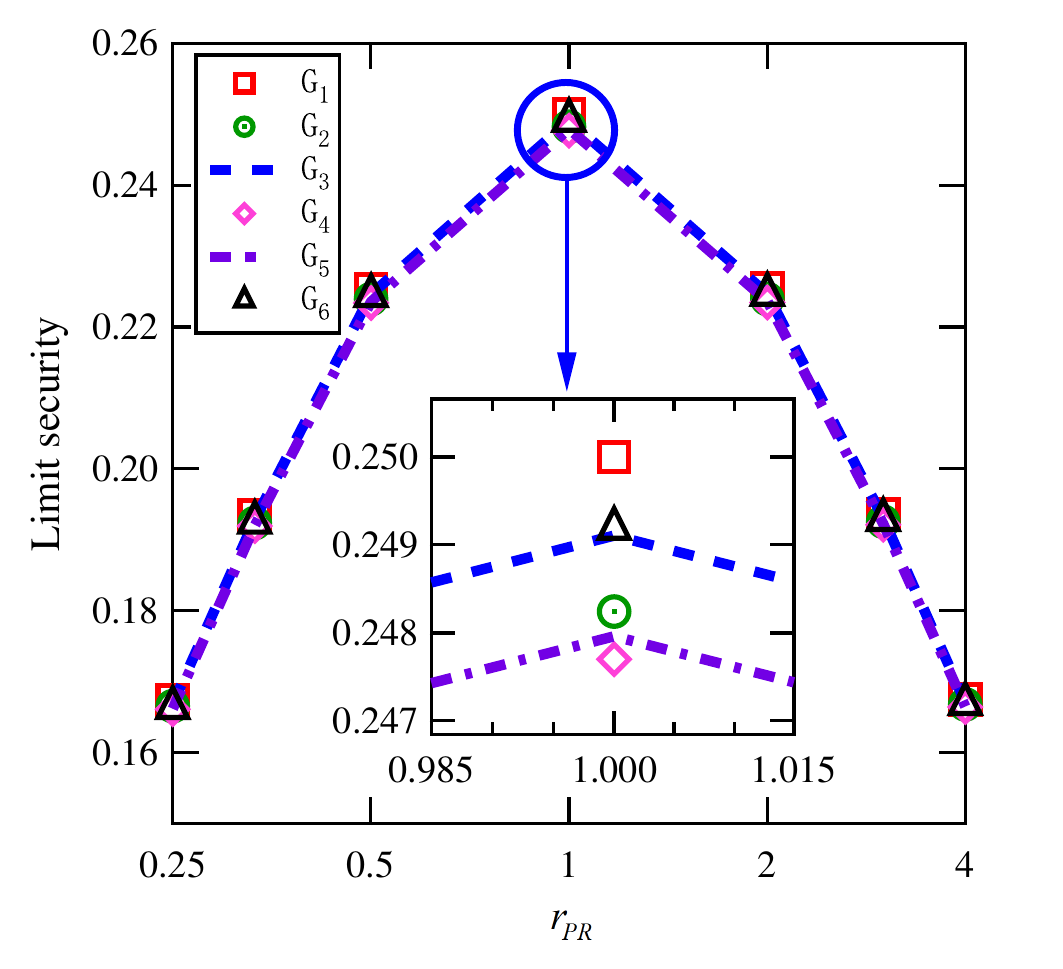}\label{a}}\\
	\subfigure[]{\includegraphics[width=0.23\textwidth,height=3.2cm]{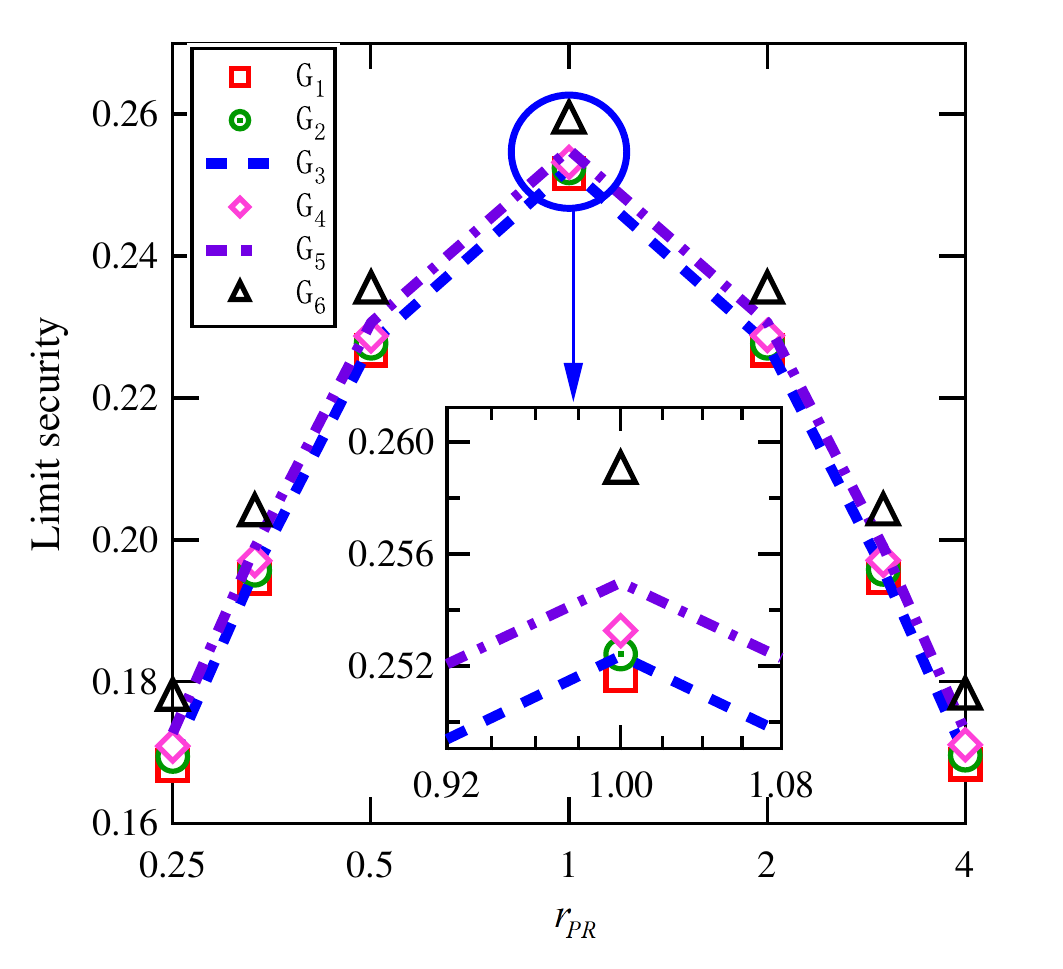}\label{a}}
	\subfigure[]{\includegraphics[width=0.23\textwidth,height=3.2cm]{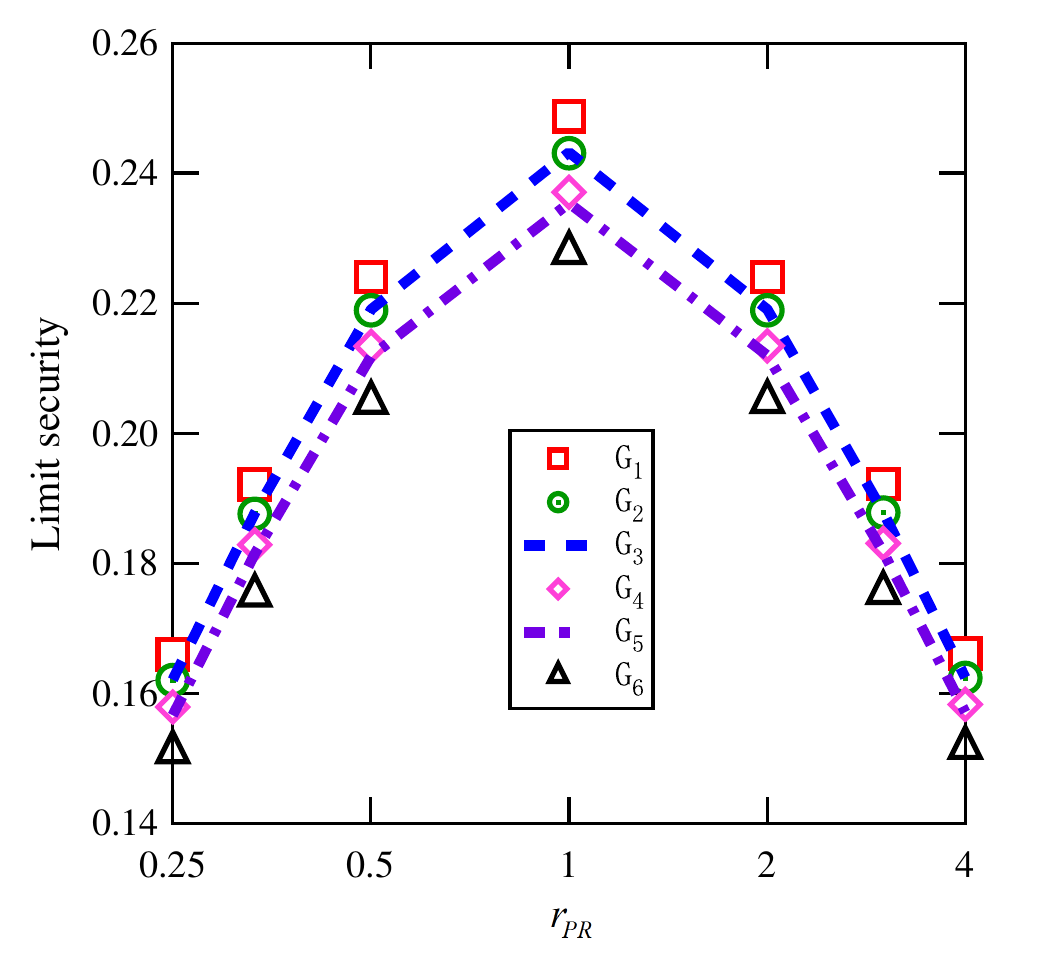}\label{a}}
	\subfigure[]{\includegraphics[width=0.23\textwidth,height=3.2cm]{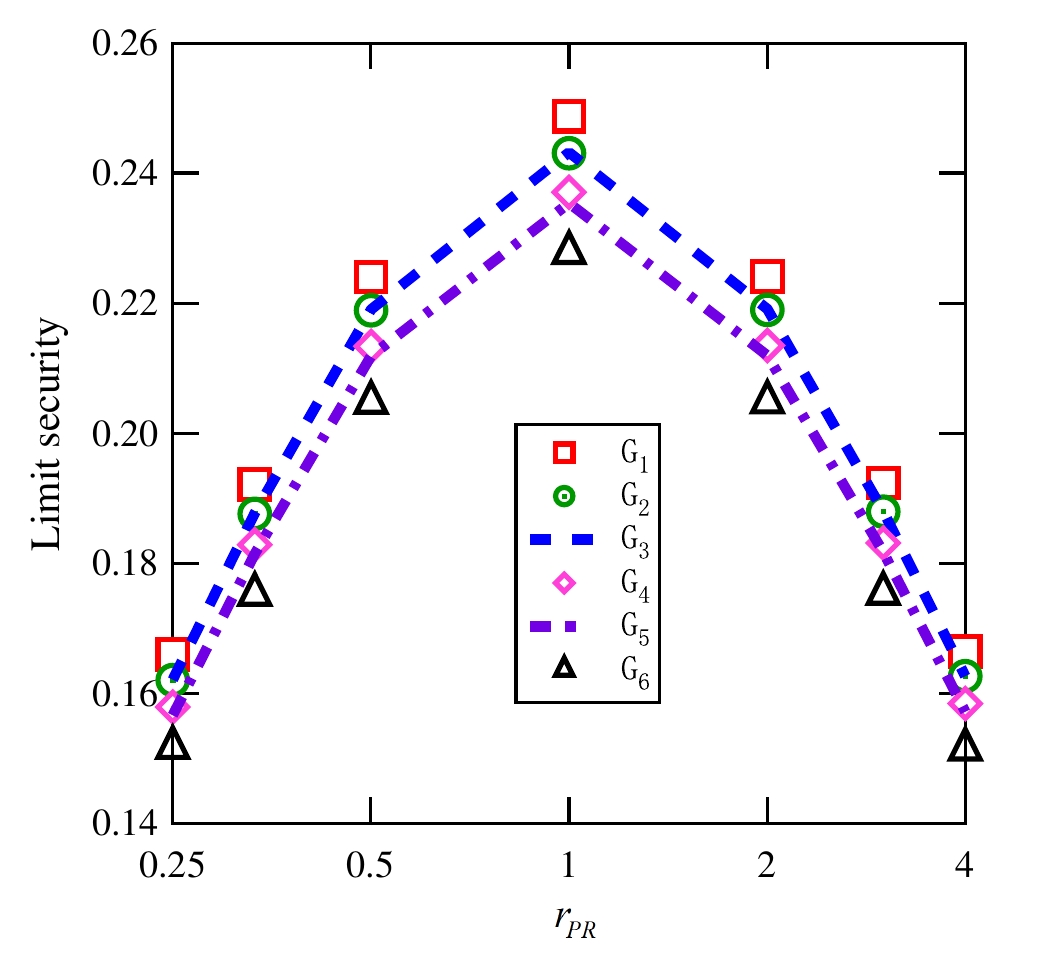}\label{a}}
	\subfigure[]{\includegraphics[width=0.23\textwidth,height=3.2cm]{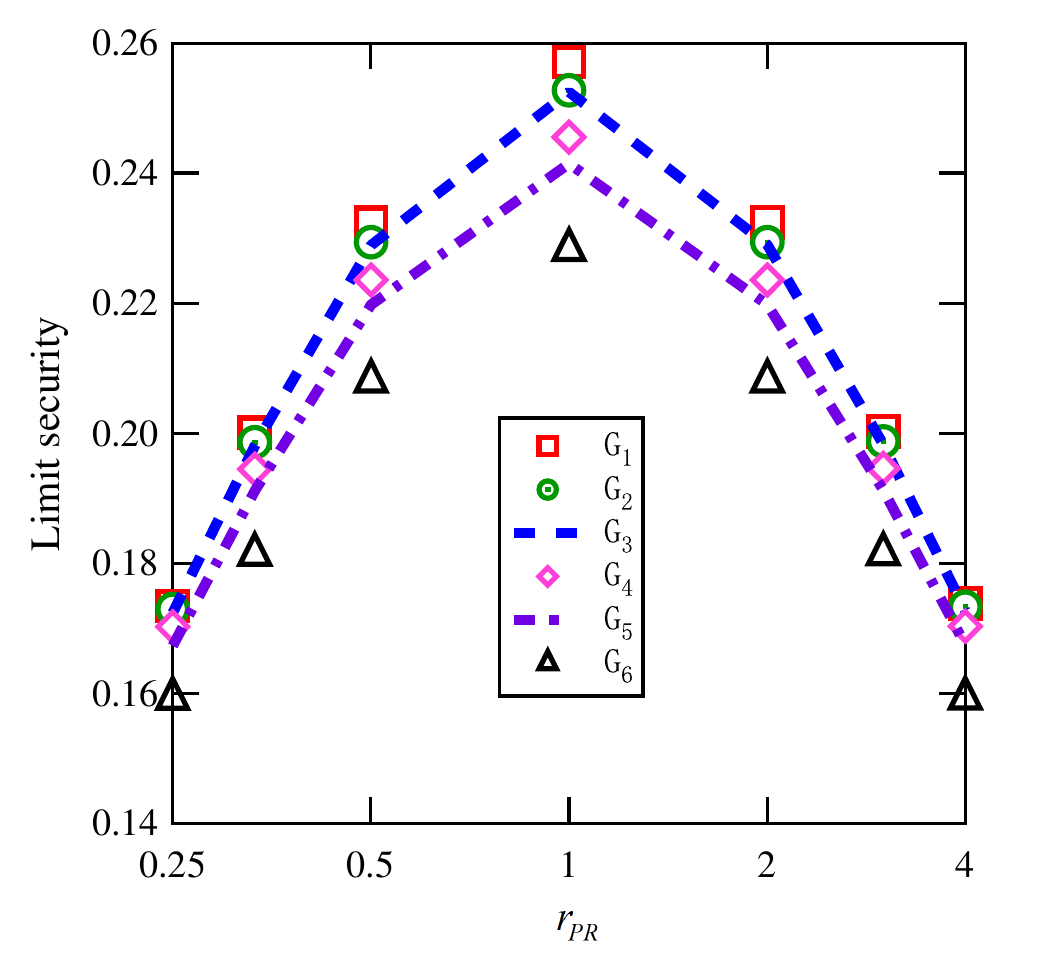}\label{a}}
	\caption{The limit security of the cyber network for each of the 504 GSCS models, where $\alpha = 0.05$, $\beta = 0.01$, $\gamma = 1$, $\delta = 1$, $G$ varies from $G_1$ to $G_6$, $||\mathbf{x}||_1 = 1$, $||\mathbf{y}||_1 = \frac{r}{1 + r}$, $||\mathbf{z}||_1 = \frac{1}{r}$, $r \in \{\frac{1}{4}, \frac{1}{3}, \frac{1}{2}, 1, 2, 3, 4\}$, with (a) uniform $\mathbf{x}$, $\mathbf{y}$ and $\mathbf{z}$; (b) uniform $\mathbf{x}$ and $\mathbf{y}$, degree-first $\mathbf{z}$; (c) uniform $\mathbf{x}$ and $\mathbf{z}$, degree-first $\mathbf{y}$; (d) uniform $\mathbf{x}$, degree-first $\mathbf{y}$ and $\mathbf{z}$; (e) degree-first $\mathbf{x}$, uniform $\mathbf{y}$ and $\mathbf{z}$; (f) degree-first $\mathbf{x}$ and $\mathbf{z}$, uniform $\mathbf{y}$; (g) degree-first $\mathbf{x}$ and $\mathbf{y}$, uniform $\mathbf{z}$; (h) degree-first $\mathbf{x}$, $\mathbf{y}$ and $\mathbf{z}$; (i) degree-last $\mathbf{x}$, uniform $\mathbf{y}$ and $\mathbf{z}$; (j) degree-last $\mathbf{x}$, uniform $\mathbf{y}$, degree-first $\mathbf{z}$; (k) degree-last $\mathbf{x}$, degree-first $\mathbf{y}$, uniform $\mathbf{z}$; (l) degree-last $\mathbf{x}$, degree-first $\mathbf{y}$ and $\mathbf{z}$. It can be seen that, with the increase of $r_{PR}$, the limit security of a cyber network goes up first but then it goes down. Moreover, the limit security attains the maximum in the proximity of $r_{PR} = 1$.}
	\vspace{2ex}
\end{figure}

\subsection{The influence of defense resource per unit time given the ratio of the attack resource to the defense resource}

For a GSCS model, the ratio of the attack resource to the defense resource is 
\[
r_{AD} = \frac{||\mathbf{x}||_1}{||\mathbf{y}||_1 + ||\mathbf{z}||_1}. 
\]
Obviously, the limit security of a cyber network declines with $r_{AD}$. A question arises naturally: given the ratio of the attack resource to the defense resource, how about the impact of the defense resource on the limit security of a cyber network? Now, let us answer the question through simulation experiments.

\begin{expe}
	Consider 504 GSCS models, where $\alpha = 0.1$, $\beta = 0.05$, $\gamma = 0.5$, $\delta = 1$, $G$ varies from $G_1$ to $G_6$, $r_{AD} = r$, $||\mathbf{y}||_1 = ||\mathbf{z}||_1 = s$, $s \in \{2, 3, \cdots, 10\}$, $||\mathbf{y}||_1 = s_1$, with (a) $r = \frac{1}{2}$, uniform $\mathbf{x}$, $\mathbf{y}$ and $\mathbf{z}$; (b) $r = 1$, uniform $\mathbf{x}$, $\mathbf{y}$ and $\mathbf{z}$; (c) $r = 2$, uniform $\mathbf{x}$, $\mathbf{y}$ and $\mathbf{z}$; (d) $r = \frac{1}{2}$, uniform $\mathbf{x}$, degree-first $\mathbf{y}$ and $\mathbf{z}$; (e) $r = 1$, uniform $\mathbf{x}$, degree-first $\mathbf{y}$ and $\mathbf{z}$; (f) $r = 2$, uniform $\mathbf{x}$, degree-first $\mathbf{y}$ and $\mathbf{z}$; (g) $r = \frac{1}{2}$, degree-first $\mathbf{x}$, uniform $\mathbf{y}$ and $\mathbf{z}$; (h) $r = 1$, degree-first $\mathbf{x}$, uniform $\mathbf{y}$ and $\mathbf{z}$; (i) $r = 2$, degree-first $\mathbf{x}$, uniform $\mathbf{y}$ and $\mathbf{z}$; (j) $r = \frac{1}{2}$, degree-first $\mathbf{x}$, $\mathbf{y}$ and $\mathbf{z}$; (k) $r = 1$, degree-first $\mathbf{x}$, $\mathbf{y}$ and $\mathbf{z}$; (l) $r =2$, degree-first $\mathbf{x}$, $\mathbf{y}$ and $\mathbf{z}$. For each of the GSCS models, the limit security of the cyber network is shown in Fig. 4. It can be seen that the limit security of a cyber network ascends with $s$. 
\end{expe}

Many similar experiments exhibit qualitatively similar phenomena. It is concluded that, given the ratio of the attack resource to the defense resource, the limit security of a cyber network goes up with the defense resource. This result sounds a good news to the defender. Indeed, configuring more defense resource is always an effective means of protecting against APTs.

\begin{figure}[H]
	\centering
	\subfigure[]{\includegraphics[width=0.23\textwidth,height=3.2cm]{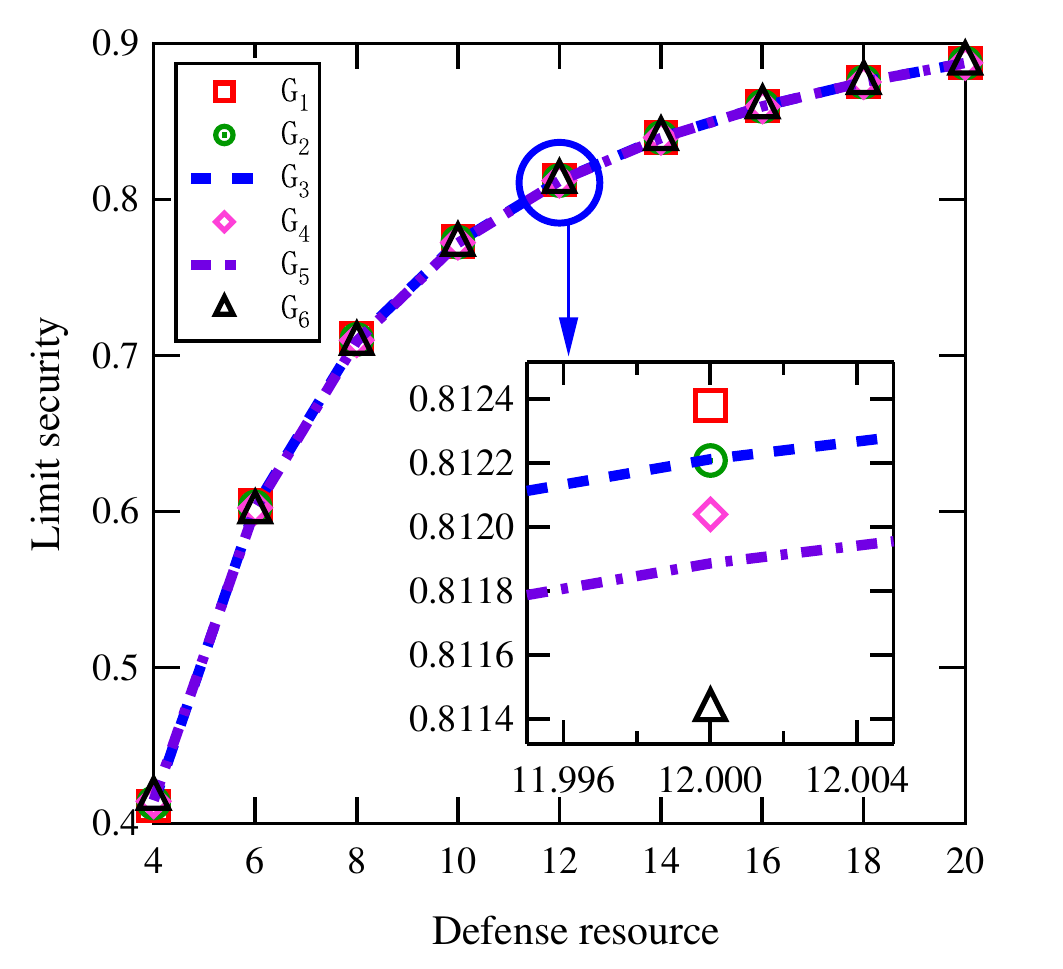}\label{a}}
	\subfigure[]{\includegraphics[width=0.23\textwidth,height=3.2cm]{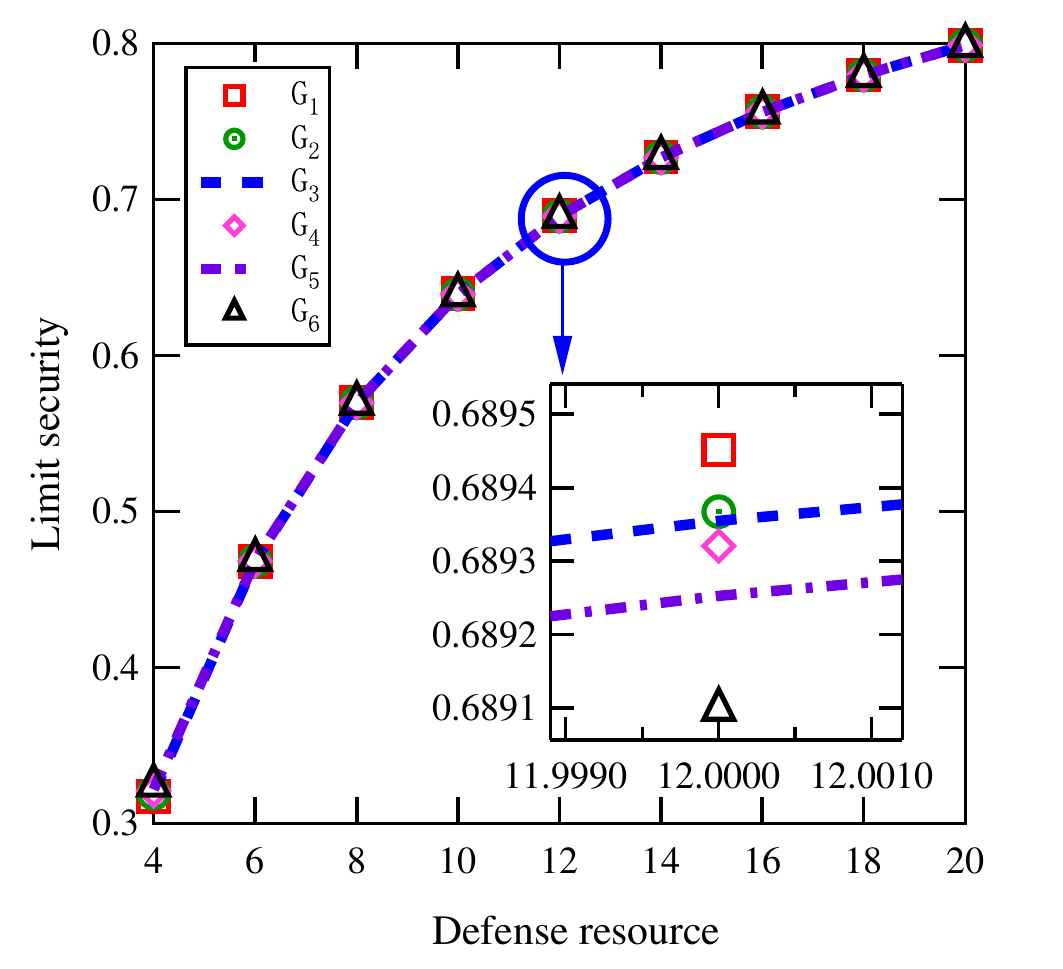}\label{a}}
	\subfigure[]{\includegraphics[width=0.23\textwidth,height=3.2cm]{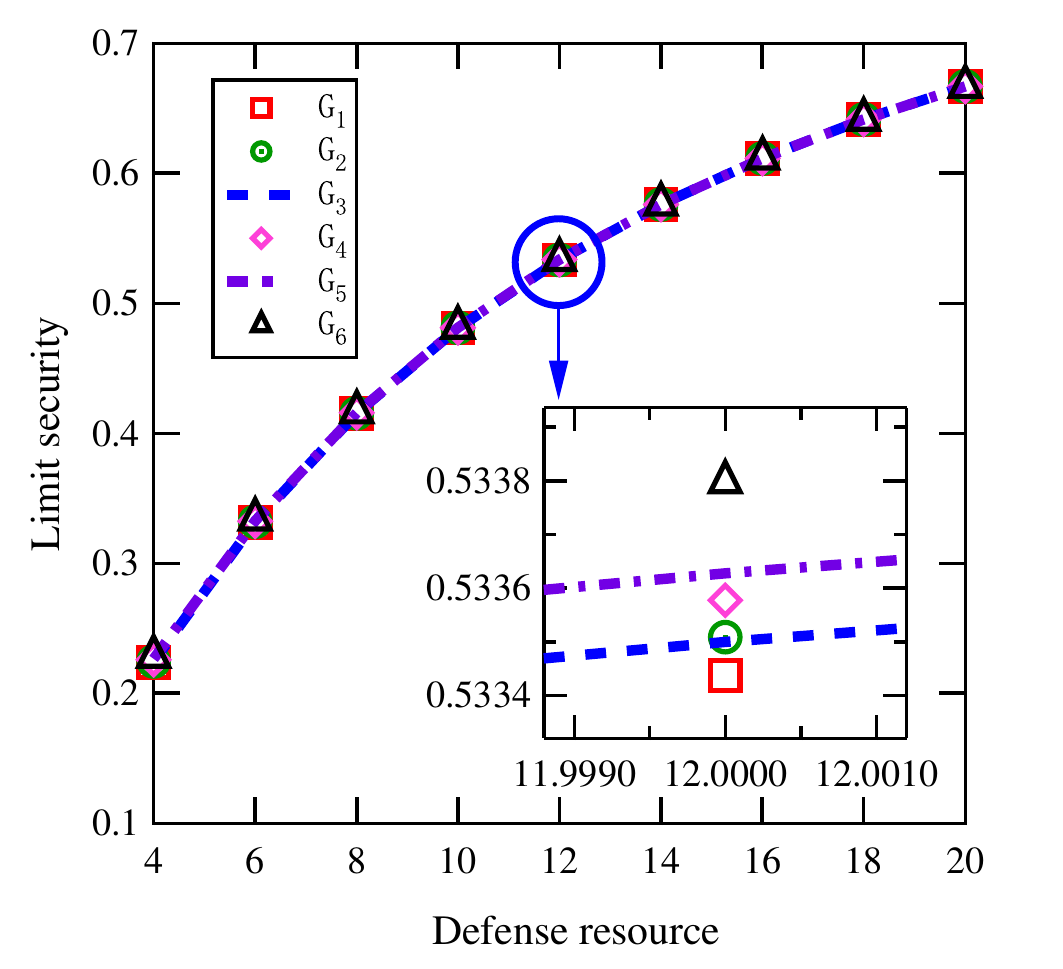}\label{a}}
	\subfigure[]{\includegraphics[width=0.23\textwidth,height=3.2cm]{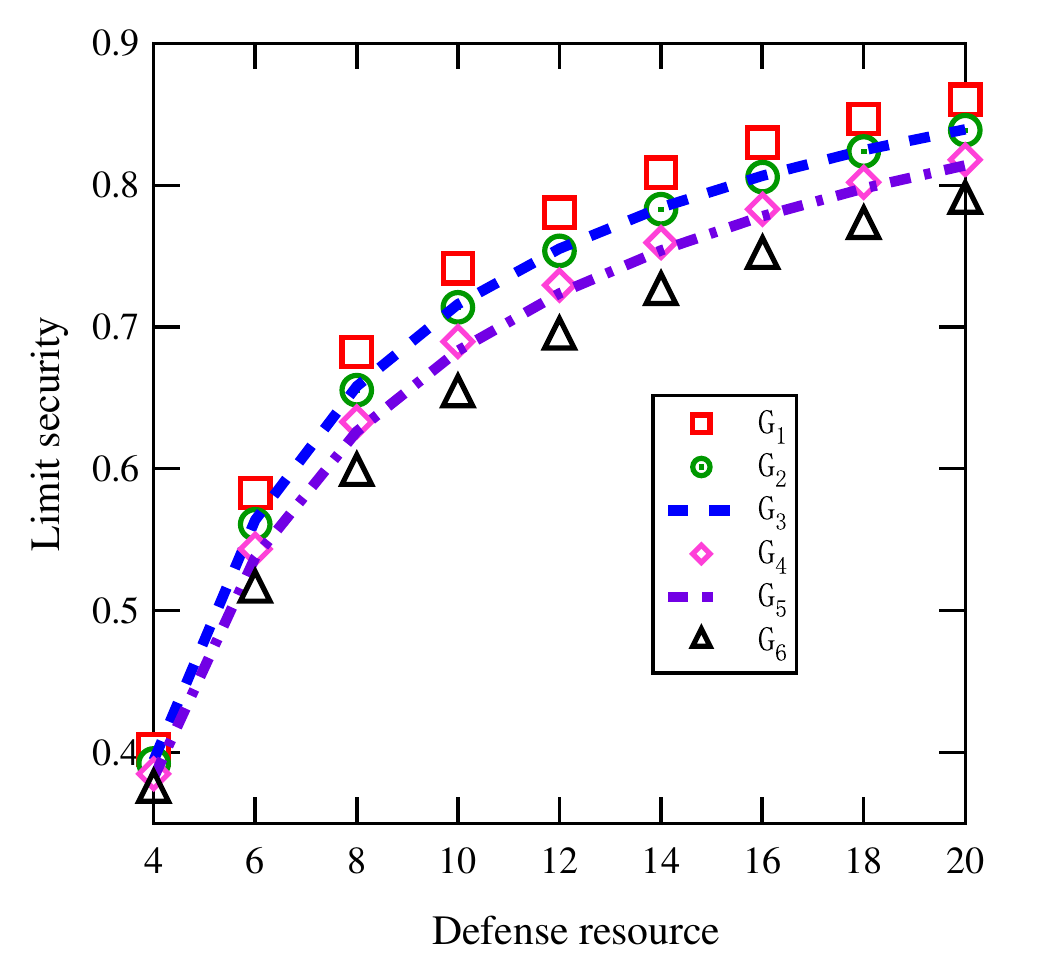}\label{a}}\\
	\subfigure[]{\includegraphics[width=0.23\textwidth,height=3.2cm]{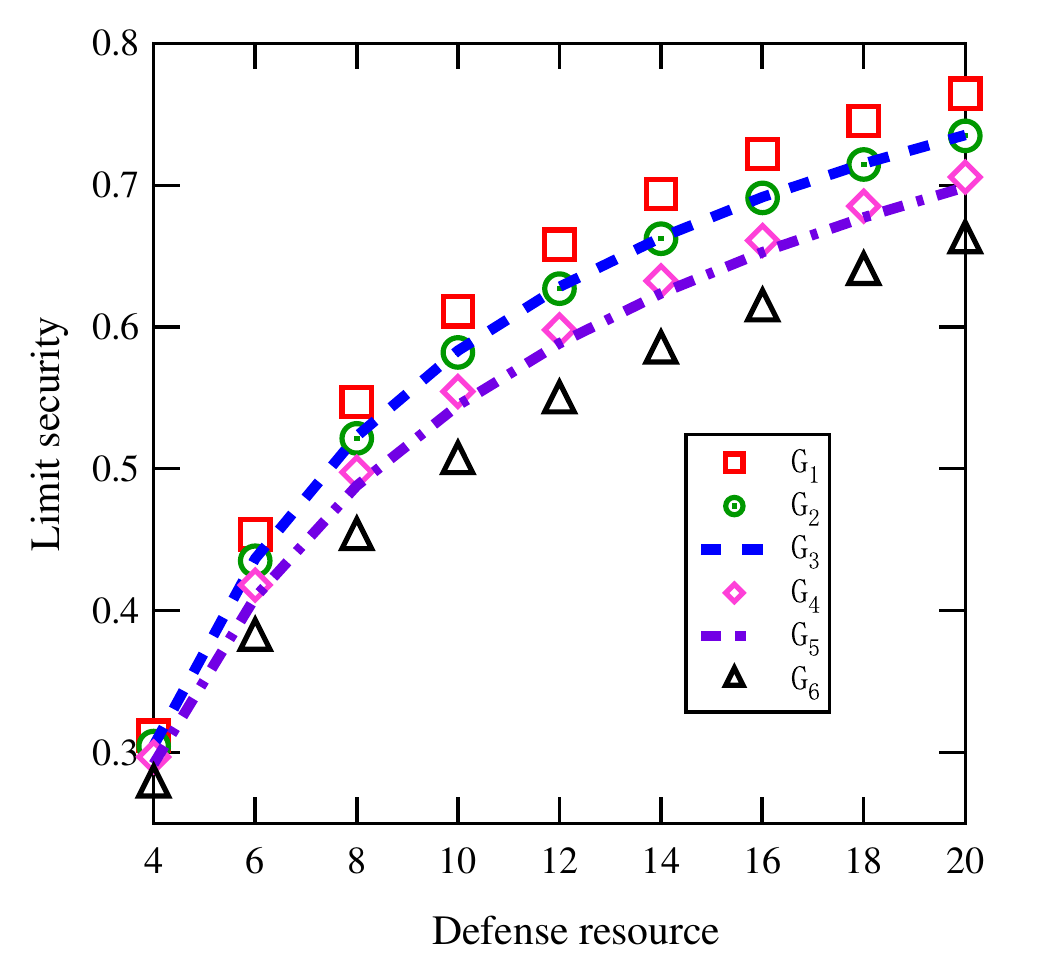}\label{a}}
	\subfigure[]{\includegraphics[width=0.23\textwidth,height=3.2cm]{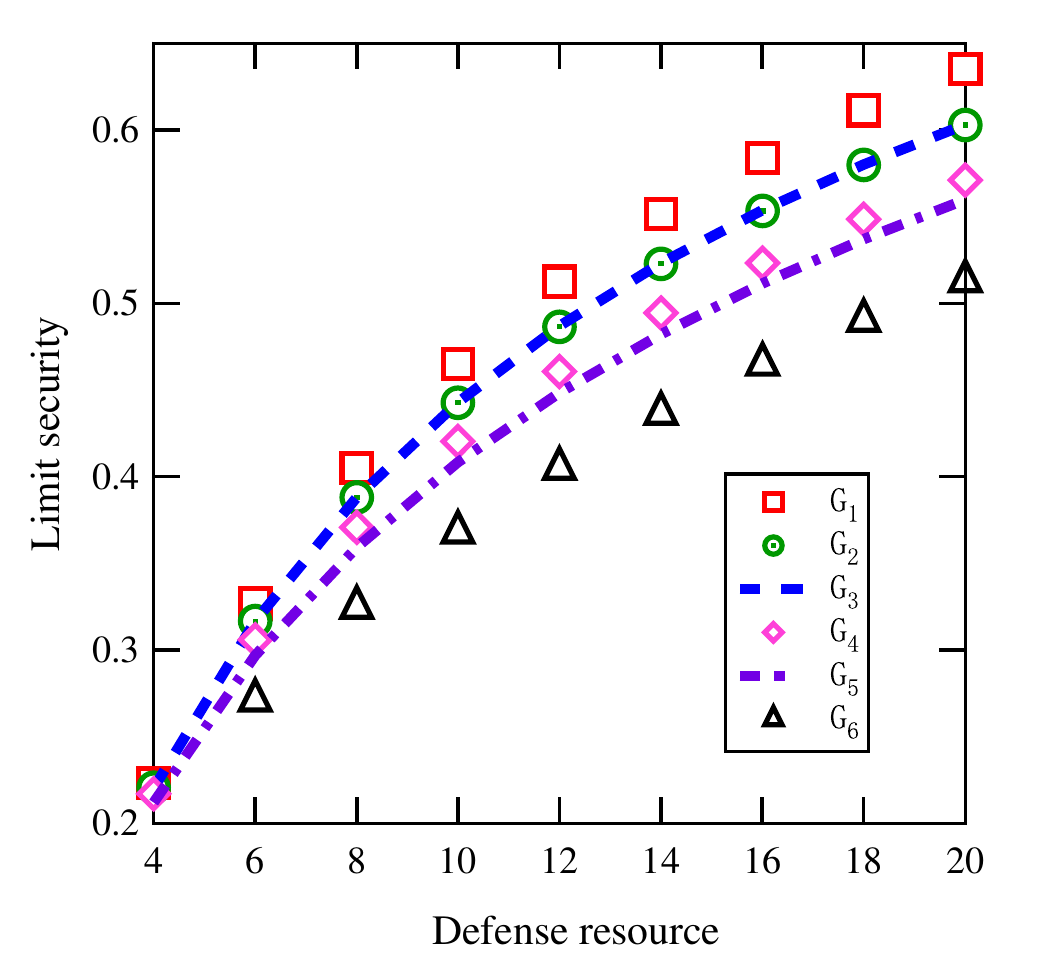}\label{a}}
	\subfigure[]{\includegraphics[width=0.23\textwidth,height=3.2cm]{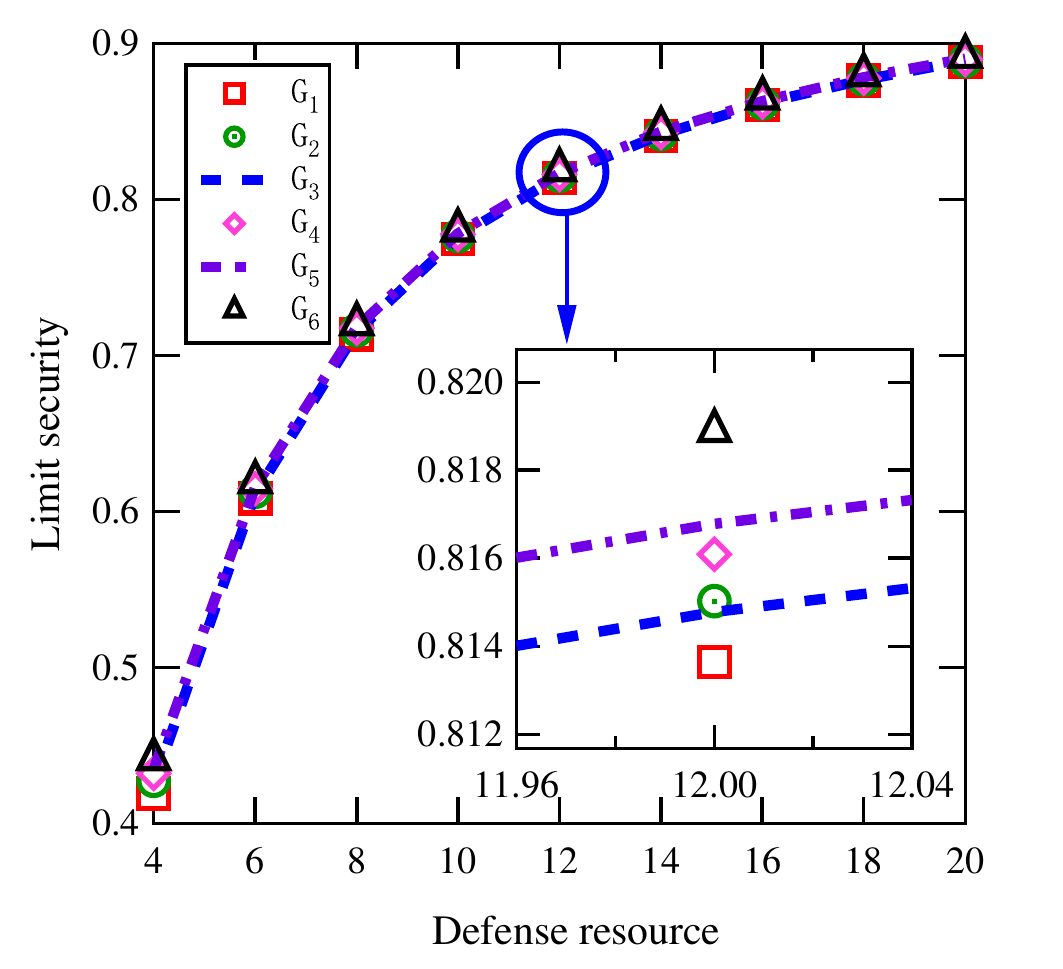}\label{a}}
	\subfigure[]{\includegraphics[width=0.23\textwidth,height=3.2cm]{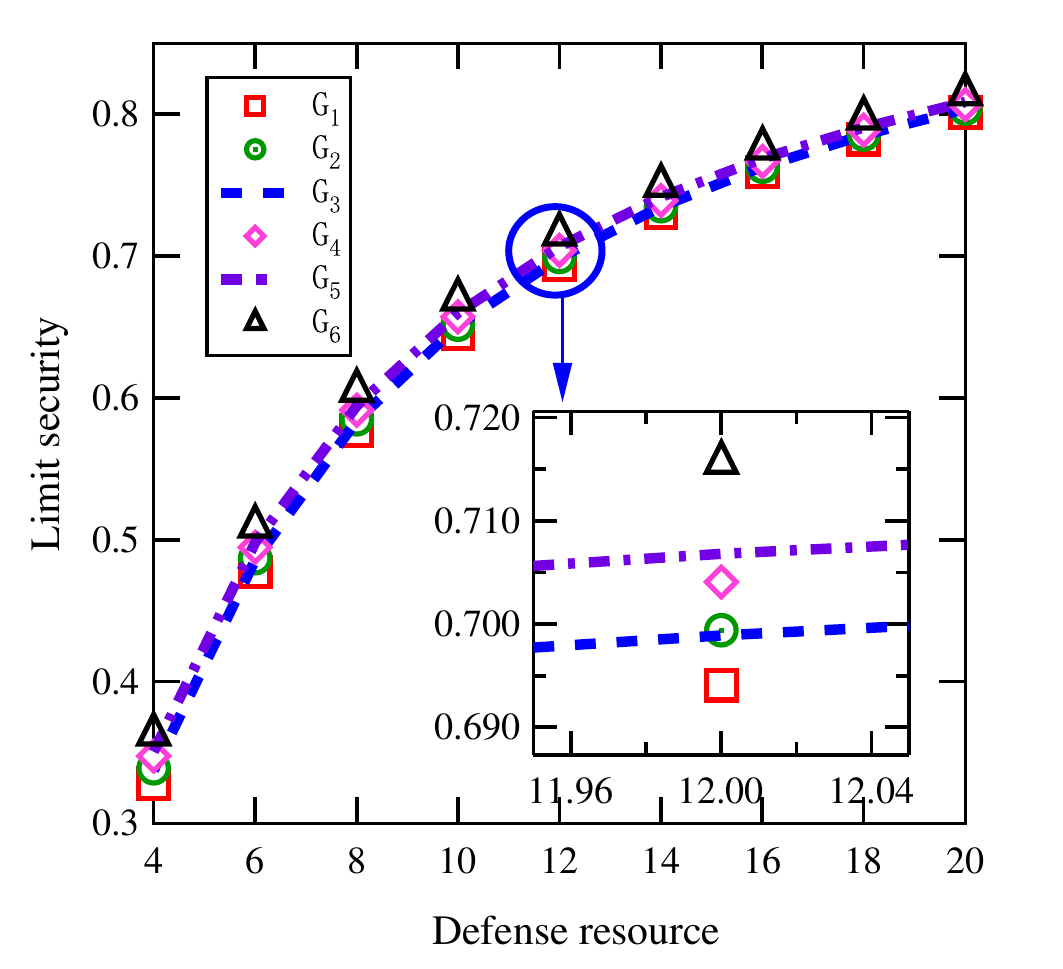}\label{a}}\\
	\subfigure[]{\includegraphics[width=0.23\textwidth,height=3.2cm]{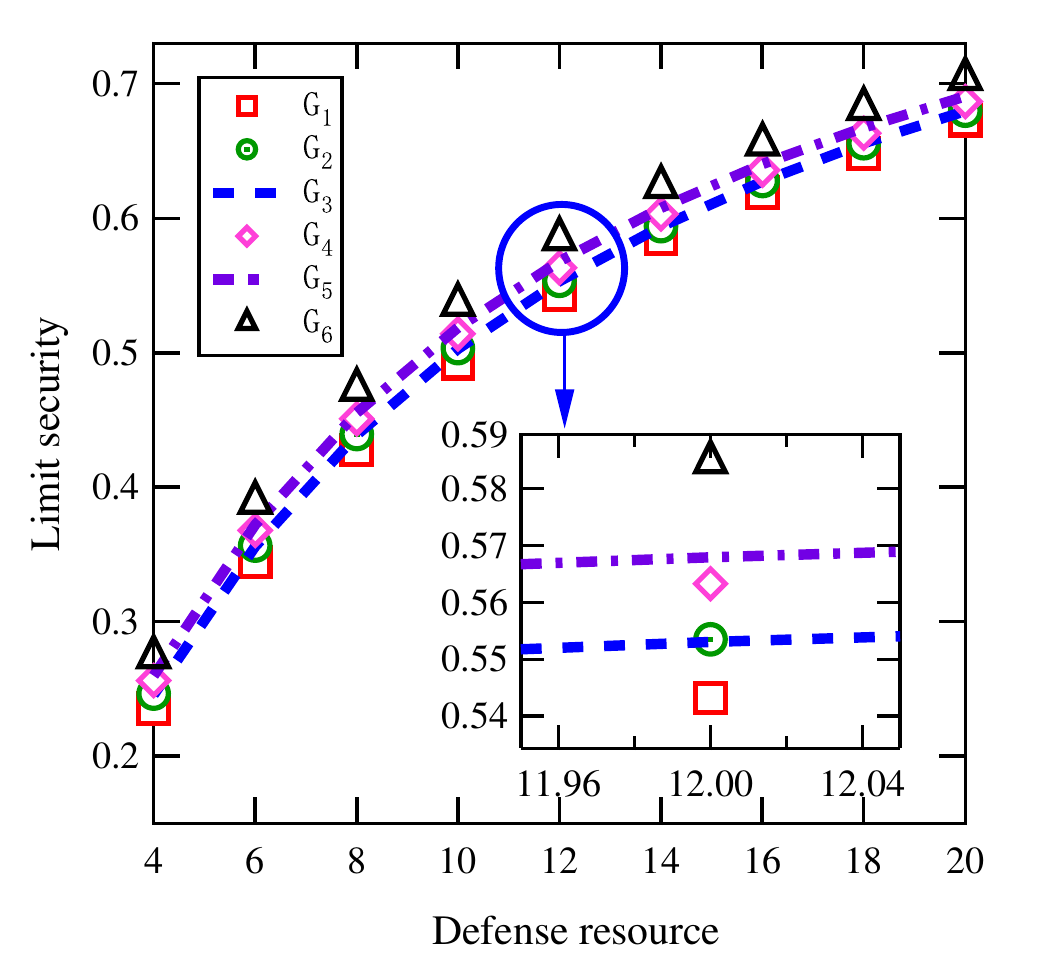}\label{a}}
	\subfigure[]{\includegraphics[width=0.23\textwidth,height=3.2cm]{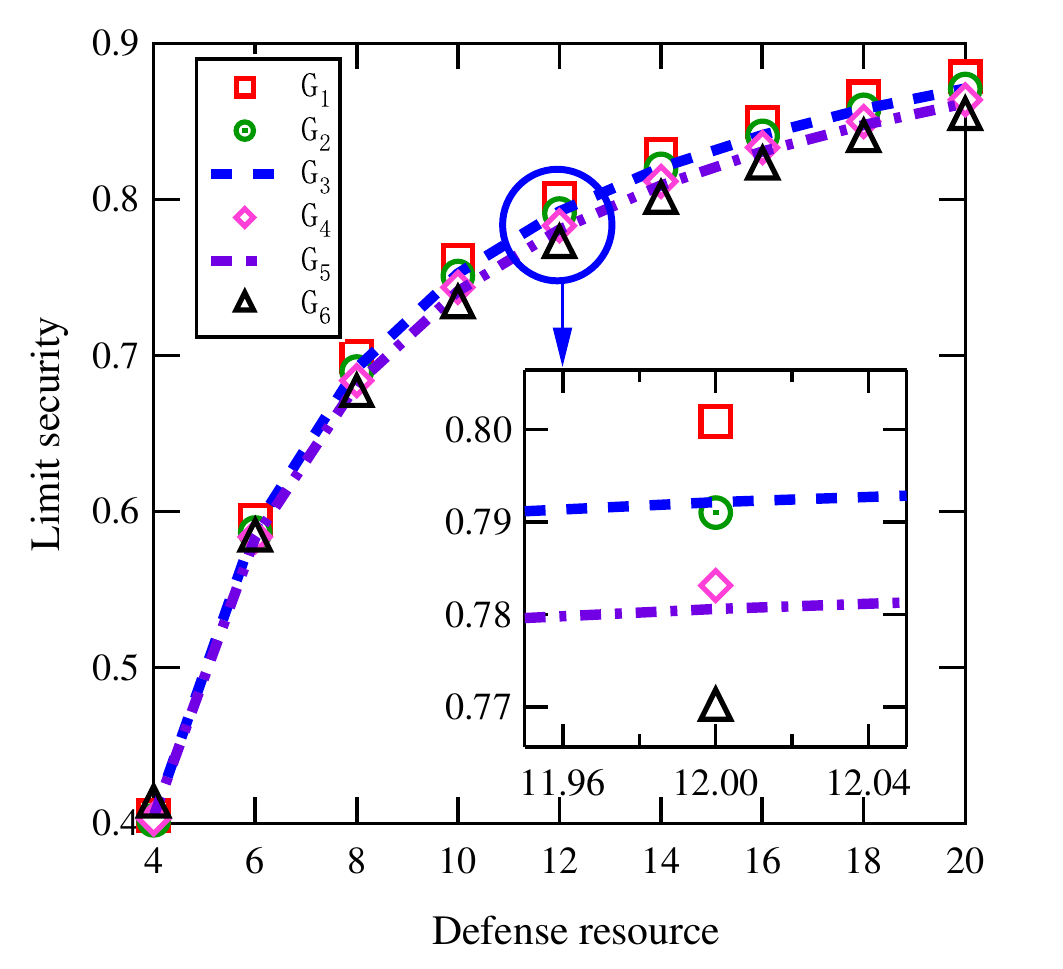}\label{a}}
	\subfigure[]{\includegraphics[width=0.23\textwidth,height=3.2cm]{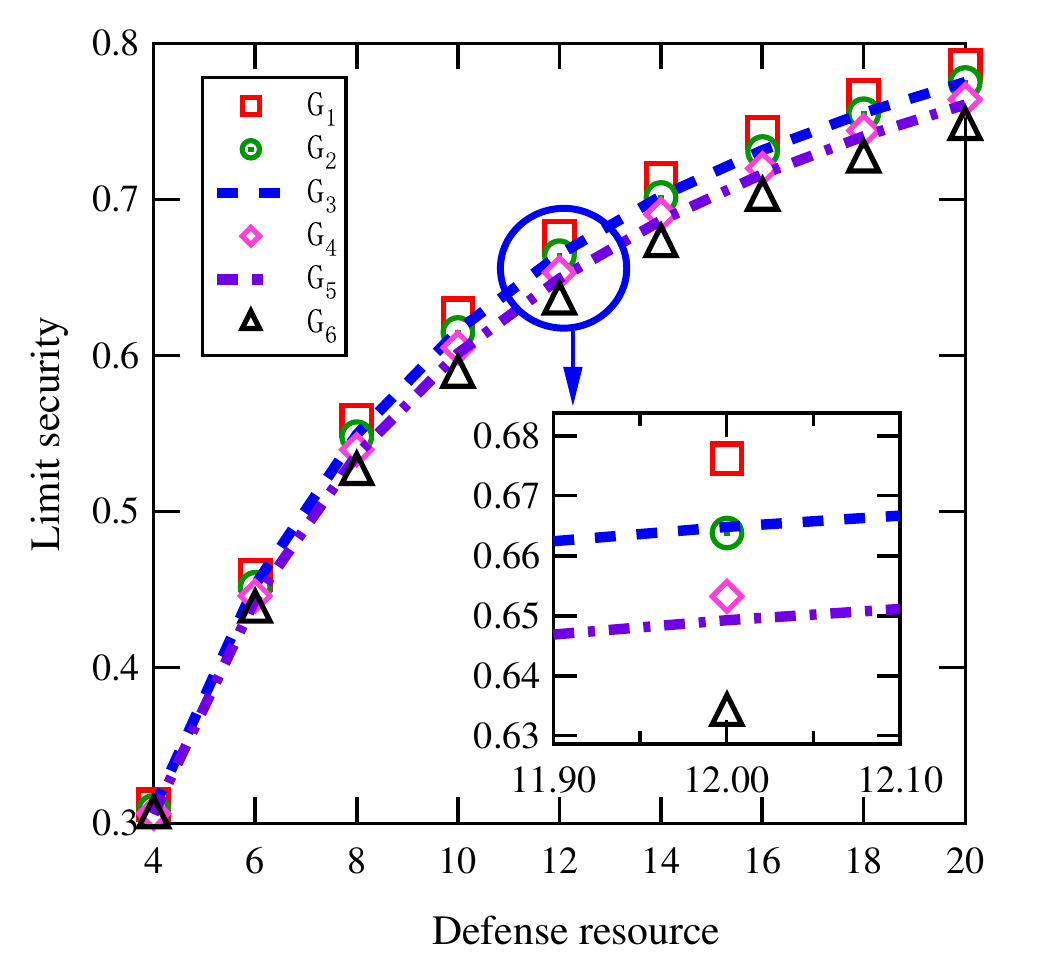}\label{a}}
	\subfigure[]{\includegraphics[width=0.23\textwidth,height=3.2cm]{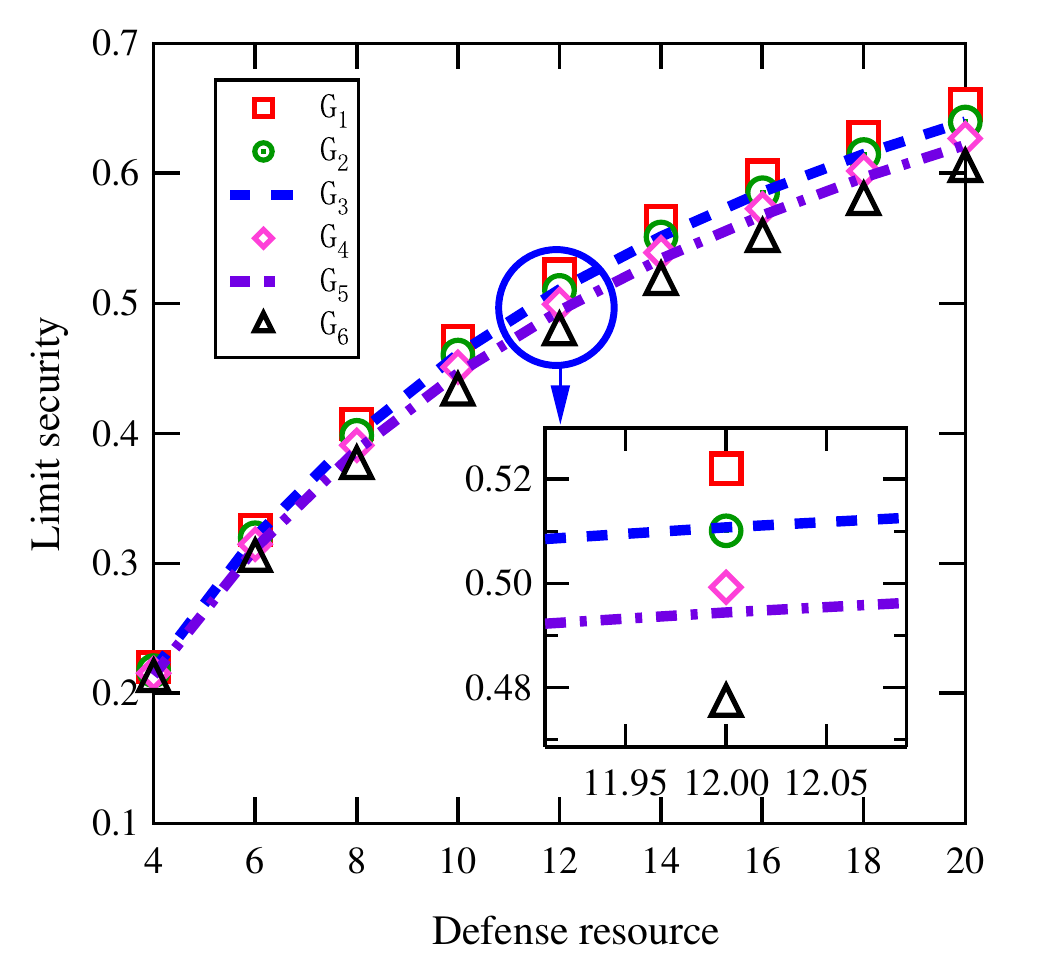}\label{a}}
	\caption{The limit security of the cyber network for each of the 504 GSCS models, where $\alpha = 0.1$, $\beta = 0.05$, $\gamma = 0.5$, $\delta = 1$, $G$ varies from $G_1$ to $G_6$, $r_{AD} = r$, $||\mathbf{y}||_1 = ||\mathbf{z}||_1 = s$, $s \in \{2, 3, \cdots, 10\}$, $||\mathbf{y}||_1 = s_1$, with (a) $r = \frac{1}{2}$, uniform $\mathbf{x}$, $\mathbf{y}$ and $\mathbf{z}$; (b) $r = 1$, uniform $\mathbf{x}$, $\mathbf{y}$ and $\mathbf{z}$; (c) $r = 2$, uniform $\mathbf{x}$, $\mathbf{y}$ and $\mathbf{z}$; (d) $r = \frac{1}{2}$, uniform $\mathbf{x}$, degree-first $\mathbf{y}$ and $\mathbf{z}$; (e) $r = 1$, uniform $\mathbf{x}$, degree-first $\mathbf{y}$ and $\mathbf{z}$; (f) $r = 2$, uniform $\mathbf{x}$, degree-first $\mathbf{y}$ and $\mathbf{z}$; (g) $r = \frac{1}{2}$, degree-first $\mathbf{x}$, uniform $\mathbf{y}$ and $\mathbf{z}$; (h) $r = 1$, degree-first $\mathbf{x}$, uniform $\mathbf{y}$ and $\mathbf{z}$; (i) $r = 2$, degree-first $\mathbf{x}$, uniform $\mathbf{y}$ and $\mathbf{z}$; (j) $r = \frac{1}{2}$, degree-first $\mathbf{x}$, $\mathbf{y}$ and $\mathbf{z}$; (k) $r = 1$, degree-first $\mathbf{x}$, $\mathbf{y}$ and $\mathbf{z}$; (l) $r =2$, degree-first $\mathbf{x}$, $\mathbf{y}$ and $\mathbf{z}$. For each of the GSCS models, the limit security of the cyber network is shown in Fig. 4. It can be seen that the limit security of a cyber network ascends with $s$.}
	\vspace{2ex}
\end{figure}

\section{Concluding remarks}

This paper is devoted to measuring the security of cyber networks under APTs. An APT-based cyber attack-defense process has been modeled as a dynamical system, which is shown to exhibit the global stability. Thereby, the limit security has been introduced as a new security metric of cyber networks. The influence of different factors on the limit security has been expounded. On this basis, some means of defending against APTs are recommended.

There are lots of open problems about APTs. In the case that the attack scheme is available, the defender must maximize the limit security over all possible defense schemes, so as to minimize the loss caused by APTs. When the attack scheme is not avaliable, the defender should furher minimize this maximized limit security over all possible attack schemes, so as to evaluate the worst-case security of the cyber network. In this work, the attack and defense schemes are both assumed to be unvaried over time. In practice, the attacker may flexibly alter the attack scheme to chase the highest profit, and the defender may flexibly change the defense scheme to maximize the security of the cyber network. In such scenarios, the evaluation of the security of cyber networks would involve optimal control theory \cite{Donald2012, YangLX2016, ZhangTR2017} or/and dynamic game theory \cite{Isaacs1999, Bressan2011}.

\section*{Acknowledgments}

This work is supported by Natural Science Foundation of China (Grant Nos. 61572006, 71301177), Sci-Tech Support Program of China (Grant
No. 2015BAF05B03), Basic and Advanced Research Program of
Chongqing (Grant No. cstc2013jcyjA1658) and Fundamental Research Funds for the Central Universities (Grant No. 106112014CDJZR008823).

\end{document}